\documentclass[onecolumn,draftcls,12pt]{IEEEtran}
\usepackage{verbatim}
\usepackage{amsfonts}
\usepackage{amssymb}
\usepackage{stfloats}
\usepackage{cite}
\usepackage{graphicx}
\usepackage{psfrag}
\usepackage{subfigure}
\usepackage{amsmath}
\usepackage{array}
\usepackage{epstopdf}
\usepackage{authblk}
\usepackage{graphicx} 
\usepackage{amsthm} 
\usepackage{lipsum}
\usepackage{verbatim} 
\usepackage{authblk}
\usepackage{mathtools}
\usepackage{cuted}
\usepackage[lined,boxed,ruled]{algorithm2e}
\usepackage{booktabs}
\usepackage{subfigure}
\usepackage{siunitx}
\usepackage{mathtools}
\usepackage{soul} 
\usepackage[]{mdframed}
\usepackage{setspace}

\newtheorem{Lemma}{Lemma}
\newtheorem{Theorem}{Theorem}

\newtheorem{Proposition}{Proposition}
\newtheorem{Remark}{Remark}

\title{On-the-Fly Communication-and-Computing for Distributed Tensor Decomposition Over MIMO Channels}
\author{}

\makeatletter
\newcommand{\removelatexerror}{\let\@latex@error\@gobble}
\makeatother

\begin{document}
\author{{Xu~Chen},~{Erik G. Larsson}, and~{Kaibin~Huang}

\thanks{X. Chen and K. Huang are with the Department of Electrical and Electronic Engineering, The University of Hong Kong, Hong Kong (Email: \{chenxu, huangkb\}@eee.hku.hk). E. G. Larsson is with the Department of Electrical Engineering (ISY), Linköping University, 58183 Linköping, Sweden (Email: \{erik.g.larsson\}@liu.se). Corresponding author: K. Huang.}}
\maketitle

\begin{abstract}
\emph{Distributed  tensor decomposition} (DTD) is a fundamental  data-analytics technique that extracts  latent important properties  from high-dimensional multi-attribute  datasets distributed over edge devices. Conventionally its wireless implementation follows a one-shot approach that first computes local results at devices using local data and then aggregates them to a server with communication-efficient techniques such as \emph{over-the-air computation} (AirComp) for global computation. Such implementation  is confronted with the issues of  limited storage-and-computation capacities and link interruption, which motivates us to propose a framework of on-the-fly communication-and-computing (FlyCom$^2$) in this work. The proposed framework enables streaming computation with low complexity by leveraging a random sketching technique and achieves progressive global aggregation through the integration of progressive uploading and  \emph{multiple-input-multiple-output} (MIMO) AirComp. To develop FlyCom$^2$, an on-the-fly sub-space estimator is designed to take real-time sketches accumulated at the server to generate online estimates for the decomposition. Its performance is evaluated by deriving both deterministic and probabilistic error bounds using the perturbation theory and concentration of measure. Both results reveal that the decomposition error is inversely proportional to the population of sketching observations received by the server. To further rein in the noise effect on the error, we propose a threshold-based scheme to select a subset of sufficiently reliable received sketches for DTD at the server. Experimental results validate the performance gain of the proposed selection algorithm and show that compared to its one-shot counterparts, the proposed FlyCom$^2$ achieves comparable (even better in the case of large eigen-gaps) decomposition accuracy besides dramatically reducing devices' complexity costs. 
\end{abstract}

\section{Introduction}
In mobile networks,  enormous amounts of data are continuously being generated by billions of edge devices. Data analytics can be performed on distributed data to support a broad range of mobile applications, ranging from e-commerce to autonomous driving to IoT sensing~\cite{Bennis6Gvision,Letaief2021JSAC}. One basic class of techniques is \emph{tensor decomposition}, which extracts a low-dimensional structure from large-scale multi-attribute data with a tensor representation (a high-dimensional counterpart of a matrix) \cite{Larsson2010,TDreview2017}. A popular technique in this class, the Tucker decomposition, is a higher-dimensional extension of the \emph{singular-value decomposition} (SVD) that has supported diverse applications such as Google's  image recognition and  Cynefin's spotting of anomalies. In mobile networks, tensor decomposition can be implemented in a  centralized manner, which requires uploading of high-dimensional data  from many devices to  a central server. However, such implementation is stymied not only by a communication bottleneck  but also by  data privacy issues~\cite{MZChen2021,NiyatoFL2020}. 

In view of these issues, we focus on \emph{distributed tensor decomposition} (DTD) that avoids direct data uploading and reduces communication overhead by distributing the computation of data tensors to the devices. A direct distributed implementation would call for parallel  iterative methods such as alternating least squares~\cite{PARAFAC2005} and stochastic gradient descent~\cite{Kijung2017,zhang2021turning} over edge devices, which, however, results in high communication overhead  due to slow convergence. On the other hand, DTD can be realized via \emph{one-shot} distributed matrix analysis techniques~\cite{chen2022analog,JF2019estimation2019,VC2021DPCA}, since the desired orthogonal factor matrices can be estimated as the principal eigenspaces of unfolding matrices of the tensor along different modes~\cite{TensorReview}. These one-shot methods improve communication efficiency at a slight cost of decomposition accuracy by following two steps: 1) computing local estimates of the desired factor matrix at each of the devices using local data; 2) uploading and aggregating the local estimates at the server to compute a global estimate. Though alleviated, the communication bottleneck still exists due to the required aggregation of high-dimensional local tensors over potentially many devices. This multi-access problem can be addressed by using a technique called \emph{over-the-air computation} (AirComp), which exploits the waveform superposition property of a multi-access channel to realize over-the-air data aggregation in one shot~\cite{GZhu2021WCM,MZChen2021}. AirComp finds applications in communication-efficient  distributed computing and learning, and has been especially popular for federated learning (see, e.g., \cite{Eldar2021TSP, Ding2020TWC,Deniz2020TSP}). 

Considering a DTD system with AirComp, this work aims to solve two open problems. The first is the prohibitive cost and latency of computation at resource-constrained devices. A traditional one-shot DTD algorithm requires each device to perform eigenvalue decomposition of a potentially high-dimensional local dataset. The resulting computation complexity increases \emph{super-linearly} with the data dimensions~\cite{SVDComplexity2019}, and the consequent latency makes it difficult for DTD to support  emerging mission-critical applications~\cite{Petar2022TimePersp}. The second problem is that the one-shot transmissions by devices are susceptible to link disruption. Specifically, a loss of connection during the transmission of high-dimensional  local principal components can render the already received partial data useless. In other words, the existing designs lack the feature of  graceful performance degradation due to fading. 

To solve these problems,  we propose the novel framework  \emph{on-the-fly communication-and-computing} (FlyCom$^2$). Underpinning this framework is the use of a technique from randomized linear algebra, \emph{randomized sketching},  that generates low-dimensional random representations, called \emph{sketches},  of a high-dimensional data sample by projecting it onto randomly generated  low-dimensional sub-spaces~\cite{StreamingTD, randomprojection2011}. This technique has been successfully used in diverse applications ranging from online data tracking~\cite{Sketching} to matrix approximation~\cite{randomprojection2011}. In FlyCom$^2$, in place  of the  traditional  high-dimensional local eigenspaces,  each device generates  a stream of low-dimensional sketches for uploading to the server. Considering a \emph{multiple-input-multiple-output} (MIMO) channel, the simultaneous transmission of local sketches is enabled by spatially multiplexed AirComp~\cite{GXZhuAirComp2019}. Upon its arrival at the server, each aggregated sketch is immediately used to improve the global tensor decomposition, giving the name of FlyCom$^2$. Since random sketches serve as independent observations of the tensor, the server can produce an estimate for tensor decomposition in every time slot based on the  sketches already received.
The FlyCom$^2$ framework addresses the above-mentioned open problems in several aspects. First, random sketching that involves matrix multiplication has much lower complexity than eigen-decomposition and helps reduce the complexity of on-device computation. Second, the DTD accuracy depends on the number of successfully received sketches and hence is robust to loss of sketches in the transmission. This gives FlyCom$^2$ a graceful degradation property in the event of link disruptions or packet losses. Third, as the principal components of a high-dimensional tensor are usually low-dimensional, the progressive DTD at the server  is shown to approach its optimal performance quickly as the number of aggregated sketches increases, thereby reining in the communication overhead. Last, the parallel streaming communication and computation in FlyCom$^2$ are more efficient than the sequential operations of the traditional one-shot algorithms due to the communication-computation separation. 

In designing the FlyCom$^2$ framework, this work makes the following key contributions. 
\begin{itemize}
    \item \emph{On-the-Fly Sub-space Detection:} One key component of  the framework is an optimal on-the-fly detector at the server to estimate the tensor's  principal eigenspace from the received, noisy (aggregated) sketches. To design a \emph{maximum likelihood} (ML) detector, a whitening technique is used to pre-process the sketches so as to yield an effective global observation, which is shown to have the covariance matrix sharing the same eigenspace as the tensor. Using the result, the ML estimation problem is formulated as a \emph{sub-space alignment problem} and is solved in closed form. It is observed from the solution that the optimal estimate of the desired principal eigenspace approaches its ground truth as the said observation's dimensionality grows (or equivalently more sketches are received). 

    \item \emph{DTD Error Analysis:} The end-to-end performance of the FlyCom$^2$-based DTD system is measured by the squared error of the estimated principal eigenspace \emph{with respect to} (w.r.t.) its ground truth. Using perturbation theory and concentration of measure, bounds are derived on both the error and its expectation. These results reveal that the error consists of one residual component contributed by non-principal components and the other component caused by random sketching. Moreover, the error is observed to be linearly proportional to the number of received sketches, validating the earlier claims on the progressive nature of the designed DTD as well as its feature of graceful degradation. This also suggests a controllable trade-off between the decomposition accuracy and communication overhead, which is useful for practical implementation.
    
    \item \emph{Threshold-Based Sketch Selection:} Removing severely channel distorted sketches from use in the sub-space detection can lead to performance improvements. This motivates the design of a sketch-selection scheme that applies a threshold on a scaling factor in MIMO AirComp that reflects the received \emph{signal-to-noise ratio} (SNR) of an aggregated sketch.  We show that such a threshold can be efficiently optimized by an enumeration method whose complexity is polynomial in the population of received sketches.
\end{itemize}

The remainder of the paper is organized as follows. Section~\ref{section:model} introduces system models and metrics. Section \ref{section:overview} gives an overview of the proposed FlyCom$^2$ framework. Then, Section~\ref{section:design}  presents the design of the on-the-fly sub-space estimator and its error analysis. The sketch-selection scheme is proposed in Section~\ref{section:selection}. Numerical results are provided in Section~\ref{section:experiment}, followed by concluding remarks in Section~\ref{section:conclusion}.
%

\section{Models, Operations, and Metrics}\label{section:model}
We consider the support of DTD in a MIMO system, as illustrated in Fig.~\ref{fig:scenario}. The relevant models, operations and metrics are described in what follows. 
\begin{figure*}[t]
	\centering
	\begin{minipage}[b]{0.76\textwidth}
		\centering
		\includegraphics[width=\textwidth]{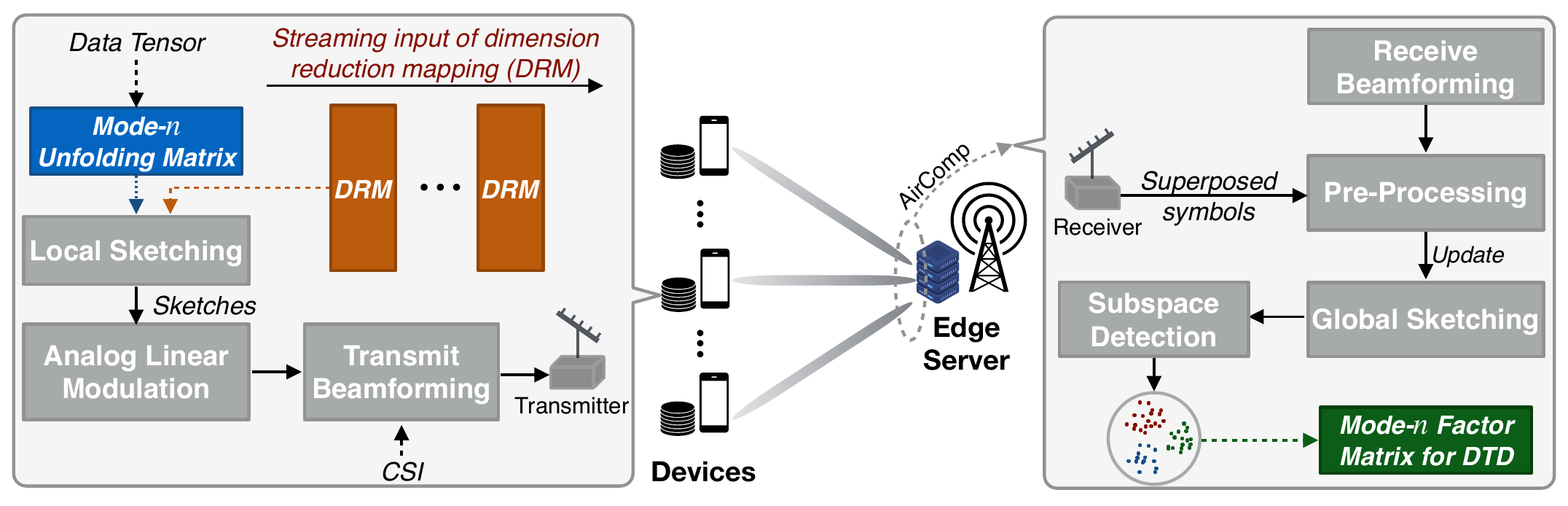}
		\vspace{-8mm}
	\end{minipage}
	\caption{On-the-fly communication-and-computing for distributed tensor decomposition.}
\label{fig:scenario}
\end{figure*}

\subsection{Distributed Tensor Decomposition} 
We consider the distributed implementation  of the popular Tucker method for tensor decomposition~\cite{TensorReview}. For ease of notation, the tensor is assumed to have $N$ modes; these modes generalize the concepts of columns and rows in matrices with the first $(N-1)$ modes corresponding to data features and mode $N$ indexing data samples. For instance, in a surveillance system, images captured by multiple cameras are expressed as local tensors with three modes indicating pixels, colors, and data sample indices, respectively. Let the samples collected by  device $k$ be represented by a local tensor $\mathcal{X}_k\in\mathbb{R}^{I_1^{(k)}\times I_2^{(k)}\cdots\times I_N^{(k)}}$, where $I_n^{(k)}$ denotes the dimensionality of mode $n$ of local tensor $k$. To simplify notation, we assume that local tensors have the same dimensions for their feature modes: $I_n^{(k)}=I_n$, $\forall k, 1\leq n\leq N-1$. Next, these local tensors are aggregated from $K$ devices to form a global tensor $\mathcal{X}\in\mathbb{R}^{I_1\times I_2\cdots\times I_N}$ with $I_N = \sum_kI_N^{(k)}$. The Tucker decomposition of $\mathcal{X}$ can be written as~\cite{TensorReview}
\begin{equation}\label{eq:Tucker}
    \mathcal{X} \approx \mathcal{G}\times_1\mathbf{U}_1\times_2\mathbf{U}_2\cdots\times_N\mathbf{U}_N\overset{\triangle}{=}\tilde{\mathcal{X}},
\end{equation}
where $\mathcal{G}\in\mathbb{R}^{r_1\times r_2\cdots\times r_N}$ represents a core tensor [that generalizes the singular values in SVD], $\mathbf{U}_n\in\mathbb{R}^{I_n\times r_n}$ is an orthogonal factor matrix corresponding to the $n$-th mode, satisfying $\mathbf{U}_n^{\top}\mathbf{U}_n = \mathbf{I}_{r_n}$ with $r_n$ ($r_n\leq I_n$) representing the number of principal dimensions, and $\times_n$ denotes the mode-$n$ matrix product~\cite{TensorReview}. In the sequel, we pursue these factor matrices $\{\mathbf{U}_n\}$ as they reveal the characteristics of the data tensor at different modes. Given $\{\mathbf{U}_n\}$, the computation of $\mathcal{G}$ is straightforward~\cite{Larsson2010}. In centralized computation with full data aggregation, $\{\mathbf{U}_n\}$ can be computed by using the \emph{higher-order SVD} approach~\cite{Larsson2010,StreamingTD}. In this approach, the tensor is first flattened along a chosen mode $n$ to yield a matrix $\mathbf{X}^{(n)}\in\mathbb{R}^{I_n\times J_n}$ with $J_n=\prod_{j=1,j\neq n}^NI_n$, termed \emph{mode-$n$ unfolding}; then the desired factor matrix is computed as $\mathbf{U}_n = [\mathbf{u}_1,\cdots,\mathbf{u}_{r_n}]$ where $\mathbf{u}_i$ is given by the $i$-th principal eigenvector of the mode-$n$ unfolding. Let this operation be represented by $\mathcal{S}_{r_n}(\cdot)$ and hence $\mathbf{U}_n = \mathcal{S}_{r_n}(\mathbf{X}^{(n)}(\mathbf{X}^{(n)})^{\top})$. 

In contrast with its centralized counterpart, DTD  computes the eigenspaces of different unfolding matrices distributively. This avoids the aggregation of raw data and preserves the data ownership~\cite{JF2019estimation2019,chen2022analog}. Considering the computation of $\mathbf{U}_n$, DTD goes through the following procedure: 1) local tensors are flattened along chosen mode $n$ to generate local unfoldings, denoted by $\{\mathbf{X}_k^{(n)}\}$; 2) devices compute low-dimensional component $\{\mathbf{S}_k\}$ from local unfoldings $\{\mathbf{X}_k^{(n)}\}$ through dimensionality reduction techniques; 3) the server gathers these local components from devices and aggregates them into a global component, denoted by $\mathbf{S}$, to yield a global estimate of the ground truth, $\mathbf{U}_n$. It is worth mentioning that the computation results $\{\mathbf{S}_k\}$ depend on a particular dimensionality reduction technique. For example,  when using \emph{principal component analysis} (PCA)~\cite{JF2019estimation2019,chen2022analog}, $\{\mathbf{S}_k\}$ are computed as the principal eigenspaces of $\{\mathbf{X}_k^{(n)}\}$ at the devices, and then the server averages them to estimate $\mathbf{U}_n$. In this work, a random sketching approach is adopted as elaborated in Section~\ref{section:overview}.

\subsection{MIMO Over-the-Air Computation}
FlyCom$^2$ builds on MIMO AirComp to  aggregate local results over the air, which is described as follows.  First, let $N_\text{r}$ and $N_{\text{t}}$, with $N_{\text{r}} \geq N_{\text{t}}$, denote the numbers of antennas at the edge server and each device, respectively. We assume perfect transmit \emph{channel state information} (CSI) as well as symbol-level and phase synchronization between the devices~\cite{GXZhuAirComp2019}. Time is slotted and then grouped to form \emph{coherence blocks} with $t$ denoting the block index. In each time slot, an $N_{\text{t}}\times 1$ vector of complex scalar symbols is transmitted over $N_{\text{t}}$ antennas. Then a coherence block spans at least $I$ symbol slots to support the transmission of an  $N_{\text{t}}\times I$ matrix. In an arbitrary block, say $t$, all edge devices transmit simultaneously
their $I\times M$ real matrices, denoted as $\{\mathbf{S}_{t,k}\}$, each of which is termed a \emph{matrix symbol}. As a result, the server receives an over-the-air aggregated  matrix symbol, $\mathbf{Y}_t$, as
\begin{equation*}
    \mathbf{Y}_t = \mathbf{A}_t\sum_{k=1}^K\mathbf{H}_{t,k}\mathbf{B}_{t,k}\mathbf{S}_{t,k}^{\top} + \mathbf{A}_t\mathbf{Z}_t,
\end{equation*}
where $\mathbf{H}_{t,k}\in \mathbb{C}^{N_{\text{r}}\times N_{\text{t}}}$ denotes the channel matrix corresponding to device $k$, $\mathbf{Z}_t$ models additive Gaussian noise with \emph{independent and identically distributed} (i.i.d.) elements of $\mathcal{CN}(0,\sigma^2)$,  $\mathbf{A}_t\in \mathbb{C}^{M\times N_{\text{r}}}$
and $\mathbf{B}_{t,k}\in \mathbb{C}^{N_{\text{t}}\times M}$ ($M\leq N_{\text{t}}$) denote receive and transmit beamforming matrices, respectively. To realize AirComp, we consider \emph{zero forcing} (ZF) transmit beamforming that inverts individual  MIMO channels~\cite{GXZhuAirComp2019}. Mathematically, conditioned on a fixed receive beamformer, transmit beamforming matrices are given as
\begin{equation}\label{eq:example:ZF}
    \mathbf{B}_{t,k} = \left(\mathbf{A}_{t}\mathbf{H}_{t,k}\right)^{H}\left(\mathbf{A}_{t}\mathbf{H}_{t,k}\mathbf{H}_{t,k}^{H}\mathbf{A}_{t}^{H}\right)^{-1}.
\end{equation}
The received matrix $\mathbf{Y}_t$ is then rewritten as
\begin{equation}\label{eq:WA}
    \mathbf{Y}_t =\sum_{k=1}^K\mathbf{S}_{t,k}^{\top}+\mathbf{A}_t\mathbf{Z}_t.
\end{equation}
In the absence of noise, the AirComp in~\eqref{eq:WA} provides the one-shot realization of the desired aggregation operation for DTD. The average transmission power of each device is constrained  to not exceed a power budget of $P$ per slot, i.e. $\forall t,k$
\begin{align}\label{eq:powerconstraint}
    \mathsf{E}\left[\Vert\mathbf{B}_{t,k}\mathbf{S}_{t,k}^{\top}\Vert_F^2\right] &= \mathsf{Tr}\left(\left(\mathbf{A}_{t}\mathbf{H}_{t,k}\mathbf{H}_{t,k}^{H}\mathbf{A}_{t}^{H}\right)^{-1}\mathsf{E}[\mathbf{S}_{t,k}^{\top}\mathbf{S}_{t,k}]\right)\nonumber\\
    &\leq IP.
\end{align}
The transmit SNR is then given by $\gamma = \frac{P}{\sigma^2}$. 

\subsection{Error Metric}
With the above-described  MIMO AirComp, a noisy version of $\mathbf{U}_n$, denoted by $\tilde{\mathbf{U}}_n$, will be computed progressively from a set of  received matrices $\{\mathbf{Y}_t\}$ (see Section~\ref{section:overview}). The $\tilde{\mathbf{U}}_n$ deviates from the ground truth due to both distributed computation and channel noise. The resulting error can be a performance metric of DTD in the wireless system. Mathematically, given $\tilde{\mathcal{X}}$ as the tensor derived from $\{\tilde{\mathbf{U}}_n\}$, the error is measured as $\Vert\mathcal{X}-\tilde{\mathcal{X}}\Vert_F^2$ that can be bounded as $ \Vert\mathcal{X}-\tilde{\mathcal{X}}\Vert_F^2 \leq \sum_{n=1}^N\Vert (\mathbf{I}_{I_n} - \tilde{\mathbf{U}}_n\tilde{\mathbf{U}}_n^{\top})\mathbf{X}^{(n)}\Vert_F^2$~\cite{StreamingTD}. This suggests that the overall error is determined by the error of independently decomposing each unfolding matrix. Hence, we define the DTD error as
\begin{equation}\label{eq:error}
    d\left(\tilde{\mathbf{U}}_n,\mathbf{X}^{(n)}\right) = \Vert (\mathbf{I}_{I_n} - \tilde{\mathbf{U}}_n\tilde{\mathbf{U}}_n^{\top})\mathbf{X}^{(n)}\Vert_F^2.
\end{equation}
%

\section{Overview of On-the-Fly Communication-and-Computing}\label{section:overview}
To support DTD over edge devices with limited computation power, we propose a FlyCom$^2$ framework as shown in Fig.~\ref{fig:scenario}. Next, we first briefly introduce the random approach exploited in FlyCom$^2$ and then explain how to use FlyCom$^2$ to support DTD.

\subsection{Data Dimensionality Reduction via Random Sketching}
 Recall that the DTD requires data dimensionality reduction on devices prior to transmission. For high-dimensional tensors, the traditional PCA technique becomes too complex for resource-constrained devices. To address this issue, we adopt a technique for random dimensionality reduction, known as \emph{random sketching}, which is simpler than PCA as it only relies on matrix multiplication and also requires a smaller number of passes over the datasets~\cite{randomprojection2011}. Specifically, given an $I\times J$ data matrix $\mathbf{X}$, random sketching uses a $J\times M$ random matrix, termed \emph{dimension reduction mapping} (DRM) and denoted by $\mathbf{\Omega}$, to map $\mathbf{X}$ to an $I\times M$ sketch matrix $\mathbf{S}$ with $J\gg M$: $\mathbf{S} = \mathbf{X}\mathbf{\Omega}$. The mapping $\mathbf{\Omega}$ can be composed of i.i.d. Gaussian elements and projects the high-dimensional $\mathbf{X}$ to random directions in a space of low dimensionality. Despite the random projection, the mutual vector distances between the rows of $\mathbf{X}$ can be approximately preserved such that the principal (column) eigenspace of the sketch, $\mathbf{S}$, constitutes a good approximation of  $\mathbf{X}$. The approximation accuracy grows as $M$ increases and becomes perfect when $M$ is equal to $J$~\cite{randomprojection2011}. Importantly, to estimate an $r$-dimensional principal eigenspace, random sketching has the complexity of $\mathcal{O}(IJM)$ and requires a single data pass of memory, as opposed to the complexity of $\mathcal{O}(\min\{I,J\}^2\times \max\{I,J\})$ and $\mathcal{O}(r)$ memory passes in PCA~\cite{SVDComplexity2019}.

\subsection{FlyCom$^2$-Based DTD}
Based on the preceding random-sketching technique, we propose the FlyCom$^2$ framework that decomposes the high-dimensional DTD into on-the-fly processing and transmission of streams of low-dimensional random sketches. Thereby, we not only overcome devices' resource constraints but also achieve a graceful reduction of DTD error as the communication time increases. Without loss of generality, we focus on the computation of the principal eigenspace $\mathbf{U}_n$ for an arbitrary data-feature mode $n$ with $n\in[1,\cdots, N-1]$. To simplify notation, the superscript $(n)$ and subscript $n$ are omitted. The detailed operations of FlyCom$^2$ are described as follows.

\subsubsection{On-the-Fly Computation at Devices} 
Each device streams a sequence of low-dimensional local sketches to the server by generating and transmitting them one by one in data packets. First, the progressive computation of local sketches at devices is introduced as follows. Let each local tensor, say $\mathcal{X}_k$ at device $k$, be flattened along the desired mode to generate the unfolding matrix $\mathbf{X}_k$. Then, in the (matrix-symbol) slot $t$, each device $k$ draws i.i.d. $\mathcal{N}(0,1)$ entries to form a $J\times M$ DRM, denoted by $\mathbf{\Omega}_{t,k}$, or retrieves it efficiently from memory~\cite{TRP_Tropp2018}. Then an $M$-dimensional local sketch for $\mathbf{X}_{k}$ can be computed as $\mathbf{S}_{t,k} = \mathbf{X}_{k} \mathbf{\Omega}_{t,k}$, which is then uploaded to the server immediately before computing the next sketch $\mathbf{S}_{t+1,k}$. This allows efficient communication-and-computation parallelization as shown in Fig.~\ref{fig:parallel}. 

\begin{figure*}[t]
	\centering
	\begin{minipage}[b]{0.6\textwidth}
		\centering
		\includegraphics[width=\textwidth]{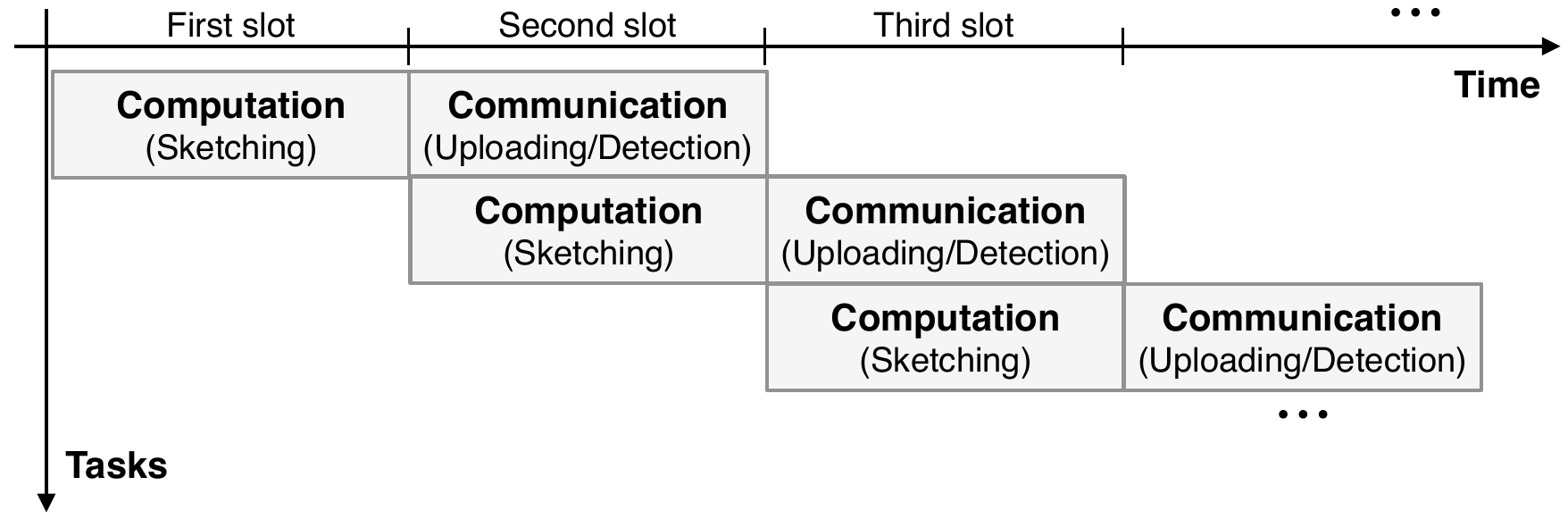}
		\vspace{-8mm}
	\end{minipage}
	\caption{Parallelization between communication and computation.}
\label{fig:parallel}
\end{figure*}

\subsubsection{On-the-Fly Global Random Sketching} 
MIMO AirComp is used for low-latency aggregation of the local sketches simultaneously streamed by devices. Local temporal sketches are progressively aggregated at the server by linearly modulating them as MIMO AirComp symbols. Consider the uploading of the $t$-th local sketches. It follows from~\eqref{eq:WA} that the matrix symbol received at the server can be written as
\begin{equation}\label{eq:receivedsymbol}
    \mathbf{Y}_t^{\top} = \sum_k\mathbf{X}_k\mathbf{\Omega}_{t,k} + \mathbf{Z}_t^{\top}\mathbf{A}_t^{\top}.
\end{equation}
To explain how to use $\mathbf{Y}_t$ in estimating the principal eigenspace of the global unfolding matrix $\mathbf{X}$, we first consider the case without channel noise, in which $\mathbf{Y}_t^{\top} = \sum_k\mathbf{X}_k\mathbf{\Omega}_{t,k}$.
Since the global tensor $\mathcal{X}$ is given by assembling local tensors along mode $N$, the corresponding global unfolding matrix, denoted by $\mathbf{X}$, is related to the local unfoldings $\{\mathbf{X}_k\}$ as
\begin{equation}\label{eq:distributed_samples}
    \mathbf{X}=[\mathbf{X}_1,\mathbf{X}_2,\cdots,\mathbf{X}_K].
\end{equation}
It follows that
\begin{equation*}
    \mathbf{Y}_t^{\top}= [\mathbf{X}_1,\cdots,\mathbf{X}_K][\mathbf{\Omega}_{t,1}^{\top},\cdots,\mathbf{\Omega}_{t,K}^{\top}]^{\top}\overset{\triangle}{=}\mathbf{X}\mathbf{F}_t,
\end{equation*}
where we define 
\begin{equation*}
    \mathbf{F}_{t} = [\mathbf{\Omega}_{t,1}^{\top},\cdots,\mathbf{\Omega}_{t,K}^{\top}]^{\top}.
\end{equation*}
As $\{\mathbf{\Omega}_{t,k}\}$ are mutually independent, $\mathbf{F}_{t}$ has i.i.d. $\mathcal{N}(0,1)$ elements and can be used as an $M$-dimensional DRM for randomly sketching $\mathbf{X}$. Therefore, in the absence of channel noise, $\mathbf{Y}_t$ gives an $M$-dimensional global sketch for $\mathbf{X}$. The dimension of the global sketch grows, thereby improving the DTD accuracy, as more aggregated local sketches are received (or equivalently as $t$ progresses), giving the name of on-the-fly global sketching.

\subsubsection{On-the-Fly Sub-space Detection at the Server}
In the case with channel noise, the server can produce an estimate of the desired principal eigenspace, $\mathbf{U}$, based on the noisy observations accumulated up to the current symbol slot. Specifically, in slot $t$,
given the current and past received matrix symbols, $\{\mathbf{Y}_{\ell}\}_{\ell\leq t}$, and the receive beamformers $\{\mathbf{A}_{\ell}\}_{\ell\leq t}$ (discussed in the sequel), the server estimates  $\mathbf{U}$ as
\begin{equation}
    \tilde{\mathbf{U}} = f(\{\mathbf{Y}_{\ell}\}_{\ell\leq t},\{\mathbf{A}_{\ell}\}_{\ell\leq t}),
\end{equation}
where the estimator $f(\cdot)$ is optimized in the sequel to minimize the DTD error in~\eqref{eq:error}.

Following the above discussion, the procedure for FlyCom$^2$-based DTD is summarized as follows.
\begin{figure}[thpb]
  \centering
\framebox{\parbox{0.98\textwidth}{To compute the principal eigenspace of the global unfolding matrix $\mathbf{X}$, initialize $t=1$, and FlyCom$^2$-based DTD repeats:
\begin{itemize}
    \setlength{\itemindent}{2.5em}
    \item[Step 1:] Each device, say device $k$, computes a local sketch using $\mathbf{S}_{t,k} = \mathbf{X}_{k} \mathbf{\Omega}_{t,k}$;
    \item[Step 2:] The server receives $\mathbf{Y}_t^{\top} = \sum_k\mathbf{X}_k\mathbf{\Omega}_{t,k} + \mathbf{Z}_t^{\top}\mathbf{A}_t^{\top}$ via MIMO AirComp;
    \item[Step 3:] The server computes an estimate of the eigenspace of $\mathbf{X}$: $\tilde{\mathbf{U}} = f(\{\mathbf{Y}_{\ell}\}_{\ell\leq t},\{\mathbf{A}_{\ell}\}_{\ell\leq t})$;
    \item[Step 4:] Set $t = t+1$;
\end{itemize}
Until $t=T$.
}}
\end{figure}

The key component of the FlyCom$^2$ framework, the on-the-fly sub-space estimator $f(\cdot)$, is designed in Section~\ref{section:design}. The performance of FlyCom$^2$-based DTD is enhanced using a sketch selection algorithm designed in Section~\ref{section:selection}.

%

\section{Optimal Sub-space Detection for FlyCom$^2$}\label{section:design}
In this section, we design the sub-space detection function of the FlyCom$^2$ framework, namely $f(\cdot)$ mentioned in the preceding section. It consists of two stages -- pre-processing of received symbols and the subsequent sub-space estimation, which are summarized in Algorithm~\ref{algo:detection} and designed in the following sub-sections. Furthermore, the resultant DTD error is analyzed.

\subsection{Pre-Processing of Received Matrix Symbols}
The pre-processing function is to accumulate received matrix symbols from slot $1$ to the current slot, $t$, and generate from them an effective matrix for the ensuing sub-space detection. The operation is instrumental for on-the-fly detection to obtain a progressive performance improvement. The design of the pre-processing takes several steps.
First, since the transmitted symbol $\mathbf{X}\mathbf{F}_{t}$ is real but the channel noise is complex, the real part of the received symbols, namely $\mathbf{Y}_t$ in~\eqref{eq:receivedsymbol}, gives an effective observation of the transmitted symbol\footnote{It is possible to transmit the coefficients of $\mathbf{X}\mathbf{F}_t$ over both the in-phase and quadrature channels, which halves air latency. The extension is straightforward (see, e.g.~\cite[Section II]{chen2022analog}) but complicates the notation without providing new insights. Hence, only the in-phase channel is used in this work.}. Let $\tilde{\mathbf{Y}}_t$ denote the effective observation in slot $t$ and $\tilde{\mathbf{Z}}_t$ the real part of $\mathbf{A}_t\mathbf{Z}_t$. It follows that
\begin{equation}\label{eq:observations}
    \tilde{\mathbf{Y}}_t = \Re\{{\mathbf{Y}_t^{\top}}\} = \mathbf{X}\mathbf{F}_{t}+\tilde{\mathbf{Z}}_t^{\top}.
\end{equation}
Second, the relation between the eigenspace of $\mathbf{X}$ and the accumulated observations up to the current slot is derived as follows. To this end, let the SVD of $\mathbf{X}$ be expressed as
\begin{equation}\label{eq:original_decomposition}
    \mathbf{X} =\mathbf{U}_{\mathbf{X}}\mathbf{\Sigma}_{\mathbf{X}}\mathbf{V}_{\mathbf{X}}^{\top},
\end{equation}
where $\mathbf{\Sigma}_{\mathbf{X}}$ comprises descending singular values along its diagonal. Then, the accumulation of the current and past observations, denoted by $\hat{\mathbf{Y}}_t = [\tilde{\mathbf{Y}}_1,\tilde{\mathbf{Y}}_2,\cdots,\tilde{\mathbf{Y}}_t]$, is a random Gaussian matrix as shown below.
\begin{Lemma}\label{Lemma:GaussianMatrix}
\emph{The accumulated aggregations, $\hat{\mathbf{Y}}_t$, can be decomposed as
\begin{equation*}
    \hat{\mathbf{Y}}_t = \mathbf{C}^{\frac{1}{2}}\mathbf{W}\mathbf{D}^{\frac{1}{2}},
\end{equation*}
where the left covariance matrix $\mathbf{C} = \mathbf{X}\mathbf{X}^{\top} + \frac{1}{2tM}\sigma^2\sum_{\ell\leq t}\mathsf{Tr}(\mathbf{A}_{\ell}^{H}\mathbf{A}_{\ell})\mathbf{I}_I$, the right covariance matrix $\mathbf{D} = \frac{\mathsf{Tr}(\mathbf{X}^{\top}\mathbf{X})\mathbf{I}_{tM} + \frac{1}{2}I\sigma^2\mathsf{diag}(\mathbf{A}_1\mathbf{A}_1^{H},\cdots,\mathbf{A}_t\mathbf{A}_t^{H})}{\mathsf{Tr}(\mathbf{X}^{\top}\mathbf{X}) + \frac{1}{2tM}I\sigma^2\sum_{\ell\leq t}\mathsf{Tr}(\mathbf{A}_{\ell}^{H}\mathbf{A}_{\ell})}$, and $\mathbf{W}$ is a random Gaussian matrix with i.i.d. $\mathcal{N}(0,1)$ entries.}    
\end{Lemma}
\begin{proof}
    See Appendix~\ref{Apdx:representation}
\end{proof}

Third, based on~\eqref{eq:original_decomposition}, the covariance matrix, $\mathbf{C}$, in Lemma~\ref{Lemma:GaussianMatrix} can be rewritten as 
\begin{equation}
    \mathbf{C} \overset{\triangle}{=}\mathbf{U}_{\mathbf{X}}\mathbf{\Lambda}\mathbf{U}_{\mathbf{X}}^{\top},
\end{equation}
where we define $\mathbf{\Lambda} = \mathbf{\Sigma}_{\mathbf{X}}^2 + \frac{1}{2tM}\sigma^2\sum_{\ell\leq t}\mathsf{Tr}(\mathbf{A}_{\ell}^{H}\mathbf{A}_{\ell})\mathbf{I}_I$. Hence, the square root, $\mathbf{C}^{\frac{1}{2}}$, is given as
\begin{equation}\label{eq:effectivecovariance}
    \mathbf{C}^{\frac{1}{2}} = \mathbf{U}_{\mathbf{X}}\mathbf{\Lambda}^{\frac{1}{2}}.
\end{equation}
\begin{Remark}[Effective Sketching with Channel Noise]\label{Remark:analogtransmission}
    \emph{According to Lemma~\ref{Lemma:GaussianMatrix} and~\eqref{eq:effectivecovariance}, the accumulated observations, $\hat{\mathbf{Y}}_t$,  gives a sketch of the matrix $\mathbf{U}_{\mathbf{X}}\mathbf{\Lambda}^{\frac{1}{2}}$ using a Gaussian DRM with the covariance of $\mathbf{D}$. The matrix $\mathbf{U}_{\mathbf{X}}\mathbf{\Lambda}^{\frac{1}{2}}$ and the unfolding matrix $\mathbf{X}$ share the eigenspace, $\mathbf{U}_{\mathbf{X}}$. Furthermore, as $\mathbf{\Lambda}$ retains the  singular values in descending order, the top-$r$ principal eigenspace of $\mathbf{U}_{\mathbf{X}}\mathbf{\Lambda}^{\frac{1}{2}}$ is identical to that of $\mathbf{X}$ for any $1\leq r\leq M$.}
\end{Remark}

Finally, according to the preceding discussion, the desired principal eigenspace of $\mathbf{X}$ can be estimated from the sketch $\hat{\mathbf{Y}}_t$. It is known that randomized sketching prefers DRMs with i.i.d. entries~\cite{randomprojection2011}. To improve the performance, $\hat{\mathbf{Y}}_t$ can be further ``whitened'' to equalize the right covariance $\mathbf{D}$. Specifically, let $\hat{\mathbf{Y}}_t$ be right-multiplied by $\mathbf{D}^{-\frac{1}{2}}$ to yield the final \emph{effective observation} in time slot $t$ as
\begin{equation}\label{eq:whitening}
\boxed{\mathbf{\Phi}_t = \hat{\mathbf{Y}}_t\mathbf{D}^{-\frac{1}{2}} = \mathbf{U}_{\mathbf{X}}\mathbf{\Lambda}^{\frac{1}{2}}\mathbf{W}.}
\end{equation}
To compute the covariance matrix $\mathbf{D}$, the server needs to acquire the value of $\mathsf{Tr}(\mathbf{X}\mathbf{X}^{\top}) = \sum_{k}\mathsf{Tr}(\mathbf{X}_k\mathbf{X}_k^{\top})$. Note that each term in the summation, say $\mathsf{Tr}(\mathbf{X}_k\mathbf{X}_k^{\top})$, relates to the covariance of transmitted symbols, $\mathsf{E}[\mathbf{S}_{t,k}^{\top}\mathbf{S}_{t,k}]$, as
\begin{equation*}
    \mathsf{E}[\mathbf{S}_{t,k}^{\top}\mathbf{S}_{t,k}] = \mathsf{E}[\mathbf{\Omega}_{t,k}^{\top}\mathbf{X}_k^{\top}\mathbf{X}_k\mathbf{\Omega}_{t,k}] = \mathsf{Tr}(\mathbf{X}_k^{\top}\mathbf{X}_k)\mathbf{I}_{M}.
\end{equation*}
Then, $\mathsf{Tr}(\mathbf{X}_k^{\top}\mathbf{X}_k)$ can be acquired at the server by one-time feedback.

\subsection{Optimal Sub-space Estimation} 
In this sub-section, the principal eigenspace of the unfolding matrix $\mathbf{X}$ with dimensions fixed as $r$, is estimated from the effective observation given in~\eqref{eq:whitening} under the ML criterion. First, using~\eqref{eq:whitening}, the distribution of the observation $\mathbf{\Phi}_t$ conditioned on $\mathbf{U}$ and $\mathbf{\Lambda}$ is given as
\begin{equation*}
    \mathsf{Pr}\left(\mathbf{\Phi}_t|\mathbf{U}_{\mathbf{X}},\mathbf{\Lambda}\right) = \frac{\exp\left(-\frac{tM}{2}\mathsf{Tr}\left(\mathbf{\Phi}_t^{\top}\mathbf{U}_{\mathbf{X}}\mathbf{\Lambda}^{-1}\mathbf{U}_{\mathbf{X}}^{\top}\mathbf{\Phi}_t\right)\right)}{(2\pi)^{ItM/2}\mathsf{det}(\mathbf{\Lambda})^{tM/2}}.
\end{equation*}
This yields the logarithm of the likelihood function, required for ML estimation, as
\begin{align}\label{eq:likelihood}
    \mathcal{L}\left(\mathbf{U}_{\mathbf{X}};\mathbf{\Phi}_t,\mathbf{\Lambda}\right) = &\ln\left(\mathsf{Pr}\left(\mathbf{\Phi}_t|\mathbf{U}_{\mathbf{X}},\mathbf{\Lambda}\right)\right),\nonumber\\
     =& -\frac{tM}{2}\mathsf{Tr}\left(\mathbf{\Phi}_t^{\top}\mathbf{U}_{\mathbf{X}}\mathbf{\Lambda}^{-1}\mathbf{U}_{\mathbf{X}}^{\top}\mathbf{\Phi}_t\right) \nonumber\\
    &-\frac{ItM}{2}\ln(2\pi) - \frac{tM}{2}\ln(\mathsf{det}(\mathbf{\Lambda})).
\end{align}
Let $\mathbf{U}$ denote the desired $r$-dimensional principal components of $\mathbf{X}$, as obtained from splitting $\mathbf{U}_{\mathbf{X}} = [\mathbf{U},\mathbf{U}^{\bot}]$. It is observed from~\eqref{eq:likelihood} that only the first term depends on the variable $\mathbf{U}$. Then letting $(\tilde{\mathbf{U}},\tilde{\mathbf{U}}_{\mathbf{X}})$ denote an estimate of $(\mathbf{U},\mathbf{U}_{\mathbf{X}})$, the ML-estimation problem can be formulated as
\begin{equation}\label{problem:MLestimation}
    \begin{aligned}
    	\mathop{\min}_{\tilde{\mathbf{U}}}&\ \ \mathsf{Tr}\left(\mathbf{\Phi}_t^{\top}\tilde{\mathbf{U}}_{\mathbf{X}}\mathbf{\Lambda}^{-1}\tilde{\mathbf{U}}_{\mathbf{X}}^{\top}\mathbf{\Phi}_t\right)\\
    	\mathrm{s.t.}&\ \ \tilde{\mathbf{U}}_{\mathbf{X}}^{\top}\tilde{\mathbf{U}}_{\mathbf{X}}=\tilde{\mathbf{U}}_{\mathbf{X}}\tilde{\mathbf{U}}_{\mathbf{X}}^{\top} = \mathbf{I},\\
     &\ \ \tilde{\mathbf{U}}_{\mathbf{X}}=[\tilde{\mathbf{U}},\tilde{\mathbf{U}}^{\bot}].
    \end{aligned}
\end{equation}

Despite the non-convex orthogonality constraints, the problem in~\eqref{problem:MLestimation} can be solved optimally in closed form, as follows. First, define the eigenvalue decomposition $\mathbf{\Phi}_t\mathbf{\Phi}_t^{\top} = \mathbf{Q}\mathbf{\Gamma}\mathbf{Q}^T$ with $\mathbf{Q}=[\mathbf{q}_1,\cdots,\mathbf{q}_I]$ and $\mathbf{\Gamma}=\mathsf{diag}(\gamma_1,\cdots,\gamma_I)$ with eigenvalues arranged in a descending order. Then, given $\mathbf{\Lambda}=\mathsf{diag}(\lambda_1,\cdots,\lambda_I)$ and $\mathbf{U}_{\mathbf{X}}=[\mathbf{u}_1,\cdots,\mathbf{u}_I]$, the objective function of~\eqref{problem:MLestimation} can be rewritten as
\begin{equation*}
    \mathsf{Tr}\left(\mathbf{\Lambda}^{-1}\mathbf{U}_{\mathbf{X}}^{\top}\mathbf{\Phi}_t\mathbf{\Phi}_t^{\top}\mathbf{U}_{\mathbf{X}}\right) = \sum_{i=1}^I\sum_{j=1}^I\lambda_j^{-1}\gamma_i(\mathbf{q}_i^{\top}\mathbf{u}_j)^2.
\end{equation*}
Next, define $x_{ij} = \mathbf{q}_i^{\top}\mathbf{u}_j$ and rewrite the constraints in~\eqref{problem:MLestimation} as $\sum_{i=1}^Ix_{ij}^2 = \mathbf{u}_j^{\top}\mathbf{Q}\mathbf{Q}^{\top}\mathbf{u}_j=1$ and $\sum_{j=1}^Ix_{ij}^2 = \mathbf{q}_i^{\top}\mathbf{U}\mathbf{U}^{\top}\mathbf{u}_i=1$. Without loss of optimality, such constraints can be further relaxed as $\sum_{i=1}^Ix_{ij}^2\geq 1$ and $\sum_{j=1}^Ix_{ij}^2\geq 1$. This allows the problem in~\eqref{problem:MLestimation} to be reformulated as a convex problem:
\begin{equation}\label{problem:new}
    \begin{aligned}
    	\mathop{\min}_{\{x_{ij}\}}&\ \ \sum_{i=1}^I\sum_{j=1}^I\lambda_j^{-1}\gamma_ix_{ij}^2\\
    	\mathrm{s.t.}&\ \ \sum_{l=1}^Ix_{il}^2\geq 1,\ \sum_{l=1}^Ix_{lj}^2\geq 1,\ \forall i,j.
    \end{aligned}
\end{equation}
Since $\lambda_1\geq\lambda_2\geq\cdots\geq\lambda_I$, the objective of~\eqref{problem:new} subject to the constraints is lower bounded as 
\begin{equation}
    \sum_{i=1}^I\sum_{j=1}^I\lambda_j^{-1}\gamma_ix_{ij}^2\geq \sum_{i=1}^I\lambda_i^{-1}\gamma_i.
\end{equation}
The lower bound can be achieved by letting $x_{ii} = 1$, $\forall i$ and $x_{ij} = 0$, $\forall i\neq j$. The optimal solution for~\eqref{problem:new} follows as shown below.
\begin{Proposition}
    \emph{Based on the ML criterion, in slot $t$, the optimal on-the-fly estimate of the $r$-dimensional principal components of the unfolding matrix, $\mathbf{X}$, is denoted as $\tilde{\mathbf{U}}^{\star}$ and given as
    \begin{equation}\label{eq:MLestimate}
    \tilde{\mathbf{U}}^{\star} = [\mathbf{q}_1,\cdots,\mathbf{q}_r] = \mathcal{S}_r\left(\mathbf{\Phi}_t\mathbf{\Phi}_t^{\top}\right),
\end{equation}
    where $\mathbf{\Phi}_t$ is the effective observation in slot $t$ as given in~\eqref{eq:whitening} and we recall $\mathcal{S}_r(\cdot)$ to yield the $r$-dimensional principal eigenspace of its argument.}
\end{Proposition}
\begin{Remark}[Minimum Number of FlyCom$^2$ Operations]
    \emph{For the result in~\eqref{eq:MLestimate} to hold, the dimensions of the current effective observations  $\mathbf{\Phi}_t$ should be larger than those of $\mathbf{U}$, i.e. $tM\geq r$. This implies that the FlyCom$^2$ should run at least $t\geq r/M$ rounds to enable the estimation of an $r$-dimensional principal eigenspace of the tensor.}
\end{Remark}
\begin{algorithm}[t]
\caption{On-the-Fly Sub-space Detection for FlyCom$^2$ Based DTD}
\label{algo:detection}
\textbf{Initialize:} Received in-phase matrix symbols $\{\tilde{\mathbf{Y}}_{\ell}\}_{\ell\leq t}$ in slot $t$\;
\textbf{Perform:}\\
    \begin{enumerate}
        \item[1:] \emph{Aggregation:} Aggregate all received matrix symbol $\{\tilde{\mathbf{Y}}_{\ell}\}_{\ell\leq t}$ into $\hat{\mathbf{Y}}_t = [\tilde{\mathbf{Y}}_1,\cdots,\tilde{\mathbf{Y}}_t]$;
        \item[2:] \emph{Whitening:} Compute the whitened version, $\mathbf{\Phi}_t$, of the aggregated matrix $\hat{\mathbf{Y}}_t$ by~\eqref{eq:whitening};
        \item[3:]  \emph{Sub-space extraction:} Compute the first $r$ eigenvectors of $\mathbf{\Phi}_t\mathbf{\Phi}_t^{\top}$ and aggregate them into $\tilde{\mathbf{U}}$.
    \end{enumerate}
\textbf{Output:} $\tilde{\mathbf{U}}$ used as the principal eigenspace of the unfolding matrix $\mathbf{X}$.
\end{algorithm}
\subsection{DTD Error Analysis} 
Based on the optimal sub-space detection designed in the preceding sub-section, we mathematically quantify the key feature of FlyCom$^2$ that the DTD error gracefully decreases with communication rounds. The existing error analysis for random sketching does not target distributed implementation and hence requires no communication links~\cite{randomprojection2011,StreamingTD}. The new challenge for the current analysis arises from the need to account for the distortion increased by the MIMO AirComp transmission. In what follows, we derive deterministic and probabilistic bounds on the DTD error defined in~\eqref{eq:error}. 

\subsubsection{Deterministic Error Bound}
As the unfolding matrix comprises $r$ principal components, its singular values can be represented as $\mathbf{\Sigma}_{\mathbf{X}} = \mathsf{diag}(\sigma_1,\sigma_2,\cdots,\sigma_I)$ with $\sigma_1=\cdots=\sigma_r\gg\sigma_{r+1}\geq\cdots\geq\sigma_I$, where we assume the same principal singular values following the literature (see, e.g.~\cite{StreamingTD}). 
\begin{Lemma}\label{Lemma:deviation}
    \emph{Consider the DTD of the unfolding matrix $\mathbf{X}$ in tensor decomposition that has an $r$-dimensional principal eigenspace $\mathbf{U}= [\mathbf{u}_1,\cdots,\mathbf{u}_r]$ and the singular values $\mathbf{\Sigma}_{\mathbf{X}}$. The estimation of $\mathbf{U}$ as in~\eqref{eq:MLestimate} yields the DTD error given as
    \begin{equation}\label{eq:rewrite_error}
        d(\tilde{\mathbf{U}},\mathbf{X})=\sum_{i=1}^r\sum_{j\geq r+1}(\sigma_i^2-\sigma_j^2)\langle\tilde{\mathbf{u}}_i,\mathbf{u}_j\rangle^2 + \sum_{i\geq r+1}\sigma_i^2. 
    \end{equation}
    }
\end{Lemma}
\begin{proof}
    See Appendix~\ref{Apdx:deviation}.
\end{proof}
On the right hand side, the first term, $\sum_{i=1}^r\sum_{j\geq r+1}(\sigma_i^2-\sigma_j^2)\langle\tilde{\mathbf{u}}_i,\mathbf{u}_j\rangle^2$, represents the error due to random sketching; the second term $\sum_{i\geq r+1}\sigma_i^2$ represents the residual error due to non-zero non-principal components of $\mathbf{X}$. 

Next, we make an attempt to characterize the behavior of each error term, $(\sigma_i^2-\sigma_j^2)\langle\tilde{\mathbf{u}}_i,\mathbf{u}_j\rangle^2$. Let $\tilde{\mathbf{u}}_i$ and $\mathbf{u}_j$ denote the $i$-th and $j$-th ($i\leq r<j$) eigenvectors of the sample covariance matrix $\frac{1}{tM}\mathbf{\Phi}_t\mathbf{\Phi}_t^{\top}$
and the covariance matrix $\mathbf{U}_{\mathbf{X}}\mathbf{\Lambda}\mathbf{U}_{\mathbf{X}}^{\top}$, respectively. The error in~\eqref{eq:rewrite_error} is caused by the perturbation $\mathbf{\Delta}=\frac{1}{tM}\mathbf{\Phi}_t\mathbf{\Phi}_t^{\top}-\mathbf{U}_{\mathbf{X}}\mathbf{\Lambda}\mathbf{U}_{\mathbf{X}}^{\top}$. Using this fact allows us to obtain the following desired result.

\begin{Lemma}\label{Lemma:UPofEachTerm}
    \emph{Consider a fixed realization $\mathbf{W}$ in the DRM, $\mathbf{\Phi}_t$, in~\eqref{eq:whitening} and the error term, $(\sigma_i^2-\sigma_j^2)\langle\tilde{\mathbf{u}}_i,\mathbf{u}_j\rangle^2$, in Lemma~\ref{Lemma:deviation} is upper bounded as
    \begin{equation*}
        (\sigma_i^2-\sigma_j^2)\langle\tilde{\mathbf{u}}_i,\mathbf{u}_j\rangle^2\leq \max \left\{4,\delta_{ij}^2\right\}\frac{\Vert\mathbf{\Delta}\mathbf{u}_j\Vert_2^2}{\sigma_i^2-\sigma_j^2},\quad i\leq r<j,
    \end{equation*}
    where $\delta_{ij} \overset{\triangle}{=} \frac{\min\{2|\tilde{\lambda}_i-\lambda_i|,(\sigma_i^2-\sigma_j^2)\}}{|\tilde{\lambda}_i-\lambda_j|}$ with $\lambda_i$ and $\tilde{\lambda}_i$ being the $i$-th eigenvalues of $\mathbf{U}_{\mathbf{X}}\mathbf{\Lambda}\mathbf{U}_{\mathbf{X}}^{\top}$ and $\frac{1}{tM}\mathbf{\Phi}_t\mathbf{\Phi}_t^{\top}$, respectively. }
\end{Lemma}
\begin{proof}
    See Appendix~\ref{Apdx:deterministicerror}.
\end{proof}

The upper bound in Lemma~\ref{Lemma:UPofEachTerm} suggests two scaling regions of the DTD error, namely $\delta_{ij} \geq 2$ and $\delta_{ij}<2$. Invoking the well-known Weyl's theorem (see, e.g.~\cite{WeylsTheorem}), the norm of the perturbation $\mathbf{\Delta}$ and hence the value of $\delta_{ij}$ reduce as FlyCom$^2$ progresses in time. This is aligned with the result in Fig.~\ref{subfig:eigenvalues} where the average value of $\{\delta_{ij}\}$ is observed to decrease with increasing communication time $t$. To simplify the analysis, we focus on the case of $\delta_{ij}\leq 2$, $\forall i\leq r<j$, by assuming sufficiently large $t$. In this case, the upper bound in Lemma~\ref{Lemma:UPofEachTerm}  simplifies to
\begin{equation}\label{eq:simplifiedUP}
    (\sigma_i^2-\sigma_j^2)\langle\tilde{\mathbf{u}}_i,\mathbf{u}_j\rangle^2\leq \frac{4\Vert\mathbf{\Delta}\mathbf{u}_j\Vert_2^2}{\sigma_i^2-\sigma_j^2},\quad i\leq r<j.
\end{equation}

Next, based on Lemma~\ref{Lemma:deviation} and~\eqref{eq:simplifiedUP}, the desired deterministic error bound is derived as follows.
\begin{Theorem}[Expected Error Bound]\label{Theorem:expectation}
    \emph{Given the receive beamformers $\{\mathbf{A}_{\ell}\}_{\ell\leq t}$ of FlyCom$^2$-based DTD and $\delta_{ij}\leq 2$, $\forall i\leq r<j$, the expected error can be bounded as
    \begin{equation*}
        \mathsf{E}[d(\tilde{\mathbf{U}},\mathbf{X})]\leq \frac{4}{tM}\sum_{i=1}^r\sum_{j\geq r+1}\frac{\lambda_j^2+\lambda_j\mathsf{Tr}(\mathbf{\Lambda})}{\sigma_i^2-\sigma_j^2} + \sum_{i\geq r+1}\sigma_i^2, 
    \end{equation*}
    where $\delta_{ij}$ and $\lambda_j$ follow those in Lemma~\ref{Lemma:UPofEachTerm}.}
\end{Theorem}
\begin{proof}
    See Appendix~\ref{Apdx:expectation}.
\end{proof}
\begin{figure*}[t]
    \centering
    \subfigure[Average value of $\{\delta_{ij}\}$ versus communication time]{\label{subfig:eigenvalues}\includegraphics[width=0.44\textwidth]{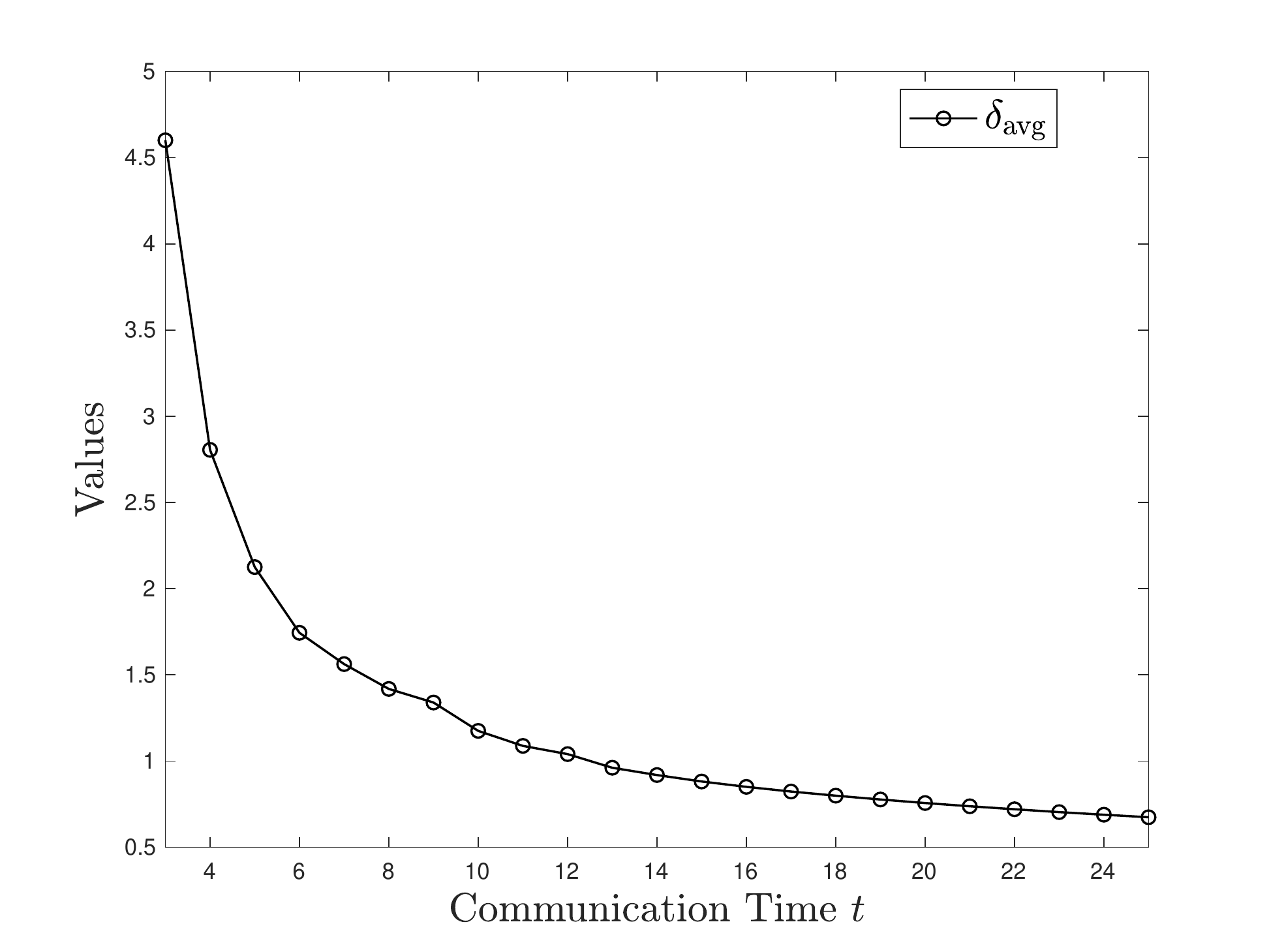}}
    \subfigure[Expected DTD error versus communication time]{\label{subfig:expectederror}\includegraphics[width=0.44\textwidth]{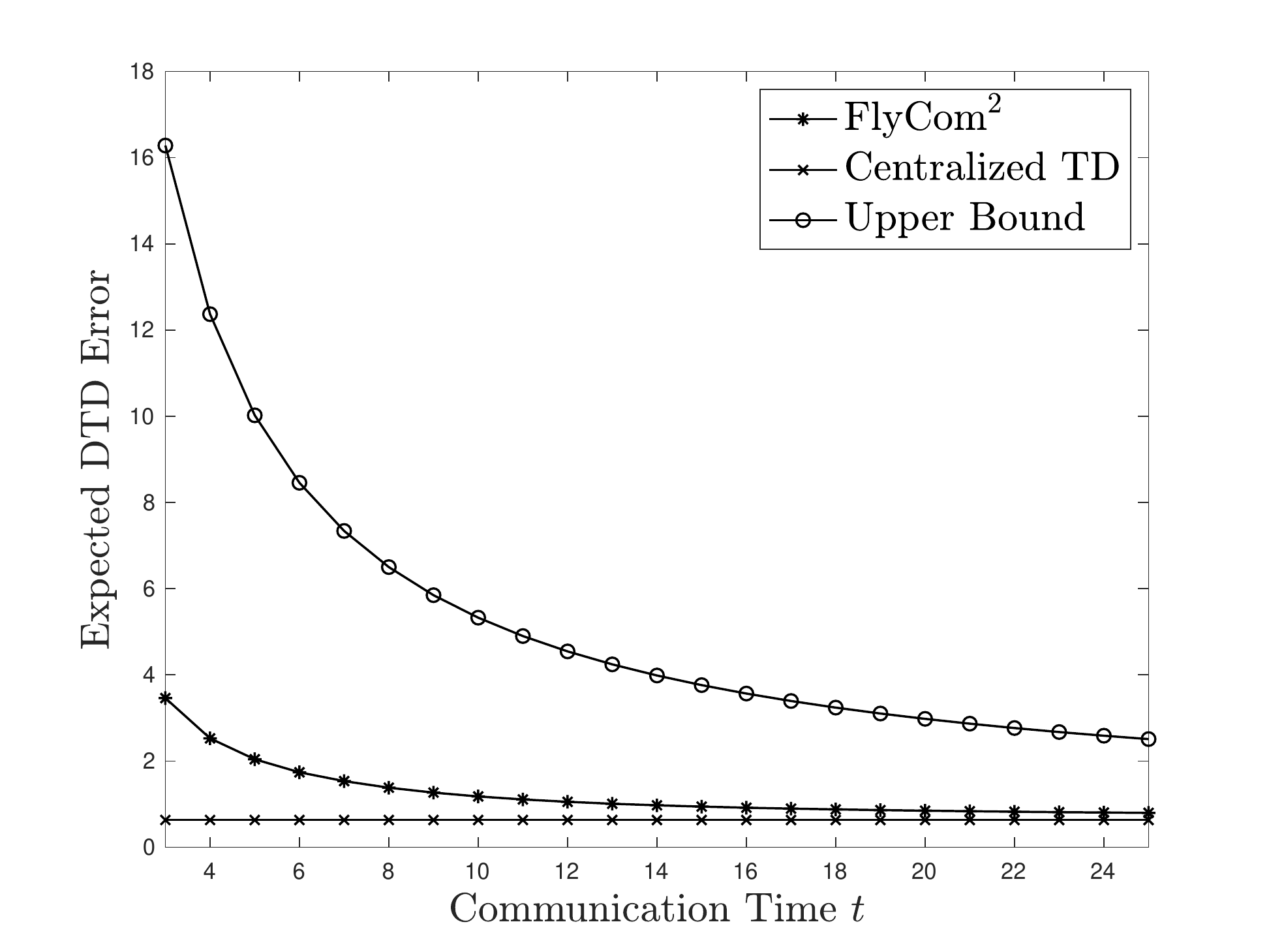}}
    \caption{Validation of theoretical results under the settings of $r = 12$, $I=100$, $\mathbf{\Sigma}_{\mathbf{X}} = \mathsf{diag}(1,\cdots,1,\frac{1}{2},\frac{1}{3},\cdots,\frac{1}{88})$, and $\mathbf{A}_t\mathbf{A}_t^{\top} =\frac{1}{10\sigma}\mathbf{I}$}
    \label{fig:validation}
\end{figure*}

The error bound in the Theorem~\ref{Theorem:expectation} is compared numerically with the exact error and that of centralized tensor decomposition in Fig.~\ref{subfig:expectederror}. One can observe the bound to capture the trend of decreasing DTD error as  $t$ progresses. In particular, it shows that under a small perturbation,
\begin{equation*}
    \mathsf{E}[d(\tilde{\mathbf{U}},\mathbf{X})]\propto \frac{1}{tM}.
\end{equation*}

\subsubsection{Probabilistic Error Bound}
We derive in the sequel a probabilistic bound on the DTD error using the method of \emph{concentration of measure}. A relevant useful result is given below.

\begin{Lemma}[McDiarmid's Inequality~\cite{McDiarmid}]\label{Lemma:McDiarmid}
    \emph{Let $g$ be a positive function on independent variables $\{W_m\}$ satisfying the bounded difference property:
    \begin{align*}
        \sup_{\{W_{M}\}_{M\neq m},W_m,W_M}|&g(\{W_{M}\}_{M\neq m},W_m) \\
        &- g(\{W_{M}\}_{M\neq m},W_M)|\leq c_m,\ \forall m,
    \end{align*}
    with constants $\{c_m\}$ and $W_M$ being i.i.d. as $W_m$. Then, for any $\epsilon>0$,
    \begin{equation*}
        \mathsf{Pr}\left[g(\{W_m\})-\mathsf{E}[g(\{W_m\})]\geq\epsilon\right]\leq \exp\left(-\frac{2\epsilon^2}{\sum_mc_m^2}\right).
    \end{equation*}}
\end{Lemma}

Using Lemma~\ref{Lemma:McDiarmid}, the desired result is obtained as shown below.
\begin{Theorem}[Probabilistic Error Bound]\label{Theorem:probability}
    \emph{Given receive beamformers $\{\mathbf{A}_{\ell}\}_{\ell\leq t}$ and $\delta_{ij}\leq 2$, $\forall i\leq r<j$, for any $\epsilon\geq 0$, the error of the FlyCom$^2$-baded DTD can be upper bounded as
    \begin{equation*}
        d(\tilde{\mathbf{U}},\mathbf{X})\leq  \frac{4(1 + \epsilon)}{tM}\sum_{i=1}^r\sum_{j\geq r+1}\frac{\lambda_j^2+\lambda_j\mathsf{Tr}(\mathbf{\Lambda})}{\sigma_i^2-\sigma_j^2} + \sum_{i\geq r+1}\sigma_i^2,
    \end{equation*}
    with the probability of at least $\left[1-e^{-\frac{\epsilon^2}{2\kappa^8}}\right]\mathsf{erf}\left(\frac{\kappa}{\sqrt{2}}\right)^{tM(I-r)}$, where $\mathsf{erf}(\cdot)$ denotes the error function defined as $\mathsf{erf}(y) = \frac{2}{\sqrt{\pi}}\int_{0}^y e^{-x^2}\mathrm{d}x$ and $\kappa\geq 1$.}
\end{Theorem}
\begin{proof}
    See Appendix~\ref{Apdx:probability}.
\end{proof}
The upper bound on the DTD error in Theorem~\ref{Theorem:probability} holds \emph{almost surely} if the constant $\epsilon$ is sufficiently large. Comparing Theorem~\ref{Theorem:expectation} and Theorem~\ref{Theorem:probability}, one can make the important observation that as the communication time ($t$) progresses, both the DTD error and its expectation vanish at the same rate of
\begin{equation}\label{eq:scalinglaw}
    \mathsf{Error} \propto \frac{1}{tM}\sum_{i=1}^r\sum_{j\geq r+1}\frac{\lambda_j^2+\lambda_j\mathsf{Tr}(\mathbf{\Lambda})}{\sigma_i^2-\sigma_j^2}.
\end{equation}
Another observation is that non-principal components of the tensor contribute to the DTD error but the effect is negligible when the eigen-gap is large.

%

\section{Optimal Sketch Selection for FlyCom$^2$}\label{section:selection}
Discarding aggregated sketches that have been transmitted under unfavourable channel conditions can improve the FlyCom$^2$ performance. This motivates us to design a sketch-selection scheme in this section.

\subsection{Threshold Based Sketch Selection}
First, we follow the approach in~\cite{GXZhuAirComp2019} to design the receive beamforming, $\{\mathbf{A}_t\}$, for MIMO AirComp. To this end, we decompose $\mathbf{A}_t$ as $\mathbf{A}_t = \eta_t\mathbf{U}_{\mathbf{A}_t}$, where the positive scalar $\eta_t$ is called a denoising factor and $\mathbf{U}_{\mathbf{A}_t}$ is an $M\times N_{\text{r}}$ unitary matrix. Following similar steps as in~\cite{GXZhuAirComp2019}, we can show that to minimize the DTD error bounds in Theorem~\ref{Theorem:expectation} and~\ref{Theorem:probability}, the beamformer component should be aligned with the channels of devices as
\begin{equation}\label{eq:receive_alignment}
    \mathbf{U}_{\mathbf{A}_t}^{\top}=\mathcal{S}_{M}\left(\frac{1}{K}\sum_{k}\lambda_{\mathbf{H}_{t,k}}\mathbf{U}_{\mathbf{H}_{t,k}}\mathbf{U}_{\mathbf{H}_{t,k}}^{\top}\right),
\end{equation}
where $\lambda_{\mathbf{H}_{t,k}}$ and $\mathbf{U}_{\mathbf{H}_{t,k}}$ denote the $N_{\text{t}}$-th eigenvalue and the first $N_{\text{t}}$ eigenvectors of  $\mathbf{H}_{t,k}\mathbf{H}_{t,k}^{\top}$, respectively. Furthermore, the denoising factor $\eta_t$ should cope with the weakest channel by being
\begin{equation}\label{eq:denoising}
    \eta_t = \max_k \frac{\mathsf{Tr}(\mathbf{X}_k^{\top}\mathbf{X}_k)}{IP} \mathsf{Tr}\left(\left(\mathbf{U}_{\mathbf{A}_t}\mathbf{H}_{t,k}\mathbf{H}_{t,k}^{H}\mathbf{U}_{\mathbf{A}_t}^{H}\right)^{-1}\right).
\end{equation}
It follows from~\eqref{eq:receive_alignment} and~\eqref{eq:denoising} that $\mathsf{Tr}(\mathbf{A}_t^{H}\mathbf{A}_t) = \eta_tM$ and $\lambda_j$ in DTD error bounds in Theorem~\ref{Theorem:expectation} and~\ref{Theorem:probability} can be expressed as
\begin{equation}\label{eq:lambda_j}
    \lambda_j=\sigma_j^2+\frac{\sigma^2}{2t}\sum_{\ell\leq t}\eta_{\ell},
\end{equation}
which shows that the error relies on only the denoising factor up to the current time slot. The result also suggests that it is preferable to select from the received sketches $\{\tilde{\mathbf{Y}}_n\}_{\ell\leq t}$ those associated with small $\eta_{\ell}$ that reflects a favorable channel condition. Naturally, we can derive a threshold based selection scheme as follows:
\begin{equation}\label{eq:selection}
    \boxed{\tilde{\mathbf{Y}}_{\ell}\ \mathrm{is\ selected\ if}\ \eta_{\ell}\leq \eta_{\mathrm{th}},\ \forall \ell\leq t,}
\end{equation}
where the threshold $\eta_{\mathrm{th}}$ is optimized in the sequel.

\subsection{Threshold Optimization}
The threshold, $\eta_{\mathrm{th}}$, in~\eqref{eq:selection}, needs to be optimized to minimize the error in~\eqref{eq:scalinglaw}. Solving the problem is hindered by that the singular values $\{\sigma_j\}$ are not available at the server in advance. We tackle this problem by designing a practical optimization scheme. To this end, we resort to using an upper bound on the DTD error as shown below.

\begin{Lemma}\label{Lemma:selection}
    \emph{Let $\tilde{M}$ denote the number of aggregated sketches selected from $\{\tilde{\mathbf{Y}}_{\ell}\}_{\ell\leq t}$ based on~\eqref{eq:selection} with the threshold $\eta_{\mathrm{th}}$, the DTD error in~\eqref{eq:scalinglaw} satisfies
    \begin{equation*}
        \frac{1}{\tilde{M}}\sum_{i=1}^r\sum_{j\geq r+1}\frac{\lambda_j^2+\lambda_j\mathsf{Tr}(\mathbf{\Lambda})}{\sigma_i^2-\sigma_j^2}\leq \frac{c}{\tilde{M}}\left[1+\frac{r\sigma^2\eta_{\mathrm{th}}}{2\sum_{k}\mathsf{Tr}(\mathbf{X}_k^{\top}\mathbf{X}_k)}\right]^2,
    \end{equation*}
    where $c$ is a constant.
    }
\end{Lemma}
\begin{proof}
    See Appendix~\ref{Apdx:selection}.
\end{proof}

Lemma~\ref{Lemma:selection} suggests that a sub-optimal threshold can be obtained by minimizing the error upper bound. Let $S$ denote a set and $|S|$ its cardinality. Then, the threshold-optimization problem can be formulated as
\begin{equation}\label{problem:threshold}
    \begin{aligned}
    	\mathop{\min}_{\eta_{\mathsf{th}}}&\ \ \frac{1}{|S|}\left[1+\frac{r\sigma^2\eta_{\mathrm{th}}}{2\sum_{k}\mathsf{Tr}(\mathbf{X}_k^{\top}\mathbf{X}_k)}\right]^2\\
    	\mathrm{s.t.}&\ \ S = \left\{\eta_{\ell}|\eta_{\ell}\leq \eta_{\mathsf{th}}, \ell\leq t\right\},
    \end{aligned}
\end{equation}
where $\eta_{\ell}$ follows the definition in~\eqref{eq:denoising}. 
One can observe that the objective function of~\eqref{problem:threshold} is a monotonically increasing function w.r.t. the variable $\eta_{\mathrm{th}}$ if $\tilde{M}$ is fixed. Such piecewise monotonicity of the objective function~\eqref{problem:threshold} renders a linear-search solution method (e.g. bisection search) infeasible, but allows the optimal solution, $\eta_{\mathrm{th}}^{\star}$, to be restricted into a finite set, say $\eta_{\mathsf{th}}^{\star} \in \left\{\eta_1,\eta_2,\cdots,\eta_t\right\}$. Then, finding $\eta_{\mathsf{th}}^{\star}$ is simple by exhausted enumeration as follows. Let $\tilde{M}_{\ell}$ denote the number of selected sketches corresponding to the threshold fixed as $\eta_{\mathsf{th}} = \eta_{\ell}$. Define
\begin{equation*}
    \ell^{\star} = \arg\min_{\ell\leq t}\frac{1}{\tilde{M}_{\ell}}\left[1+\frac{r\sigma^2\eta_{\ell}}{2\sum_{k}\mathsf{Tr}(\mathbf{X}_k^{\top}\mathbf{X}_k)}\right]^2.
\end{equation*}
The optimal threshold solving the problem in~\eqref{problem:threshold} is $\eta_{\mathsf{th}} = \eta_{\ell^{\star}}$. Two remarks are offered as follows. First, the above research for the optimal threshold has complexity linearly proportional to $t$, the population of accumulated sketches at the server. Second, the implementation of the optimization at the server requires feedback of a scalar from each device, namely $\mathsf{Tr}\left(\mathbf{X}_k^{\top}\mathbf{X}_k\right)$ from device $k$.  



%

\section{Experimental Results}\label{section:experiment}
\begin{figure*}[t]
    \centering
    \subfigure[$M= 1$]{\label{subfig:M1}\includegraphics[width=0.24\textwidth]{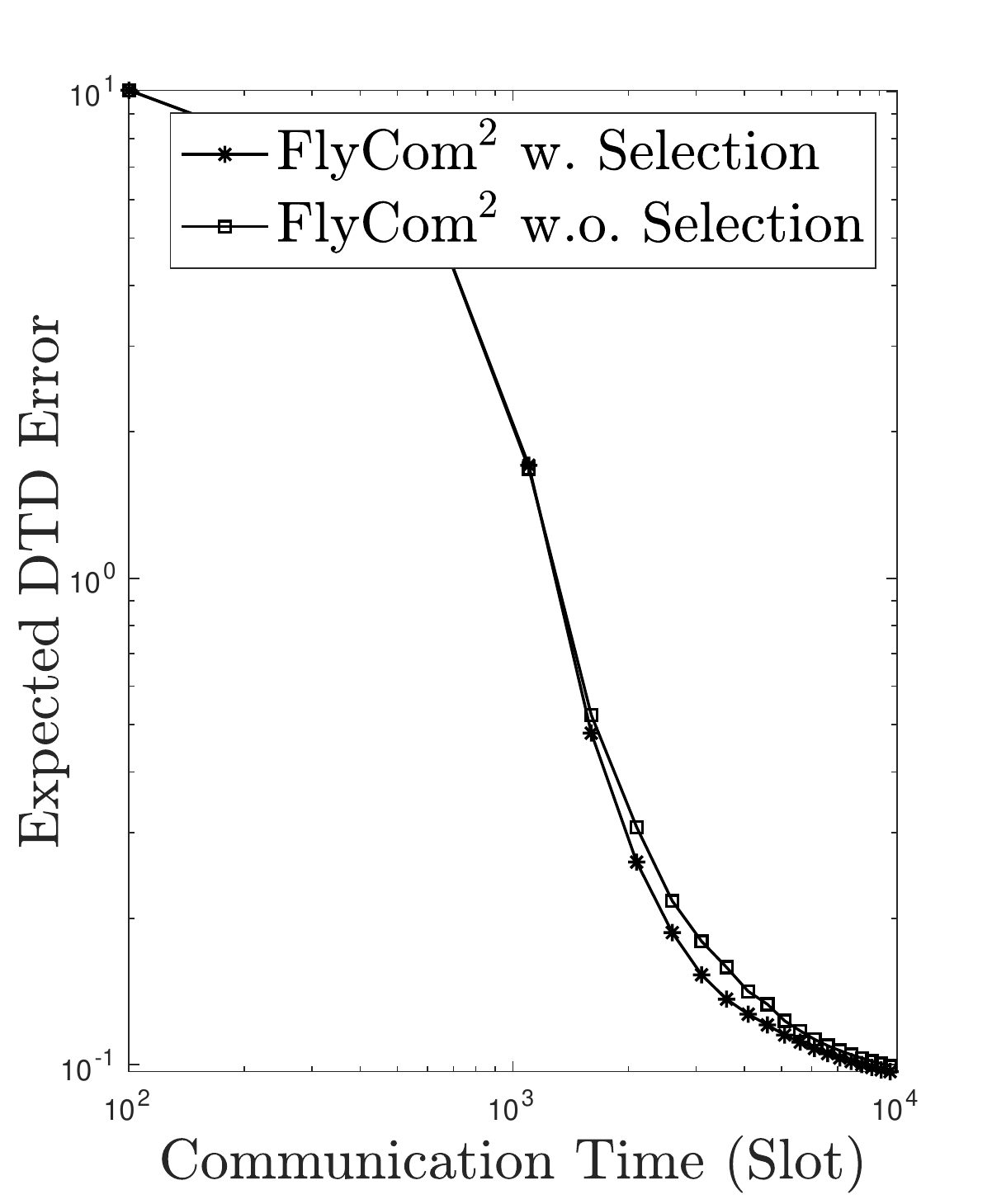}}
    \subfigure[$M= 2$]{\label{subfig:M2}\includegraphics[width=0.24\textwidth]{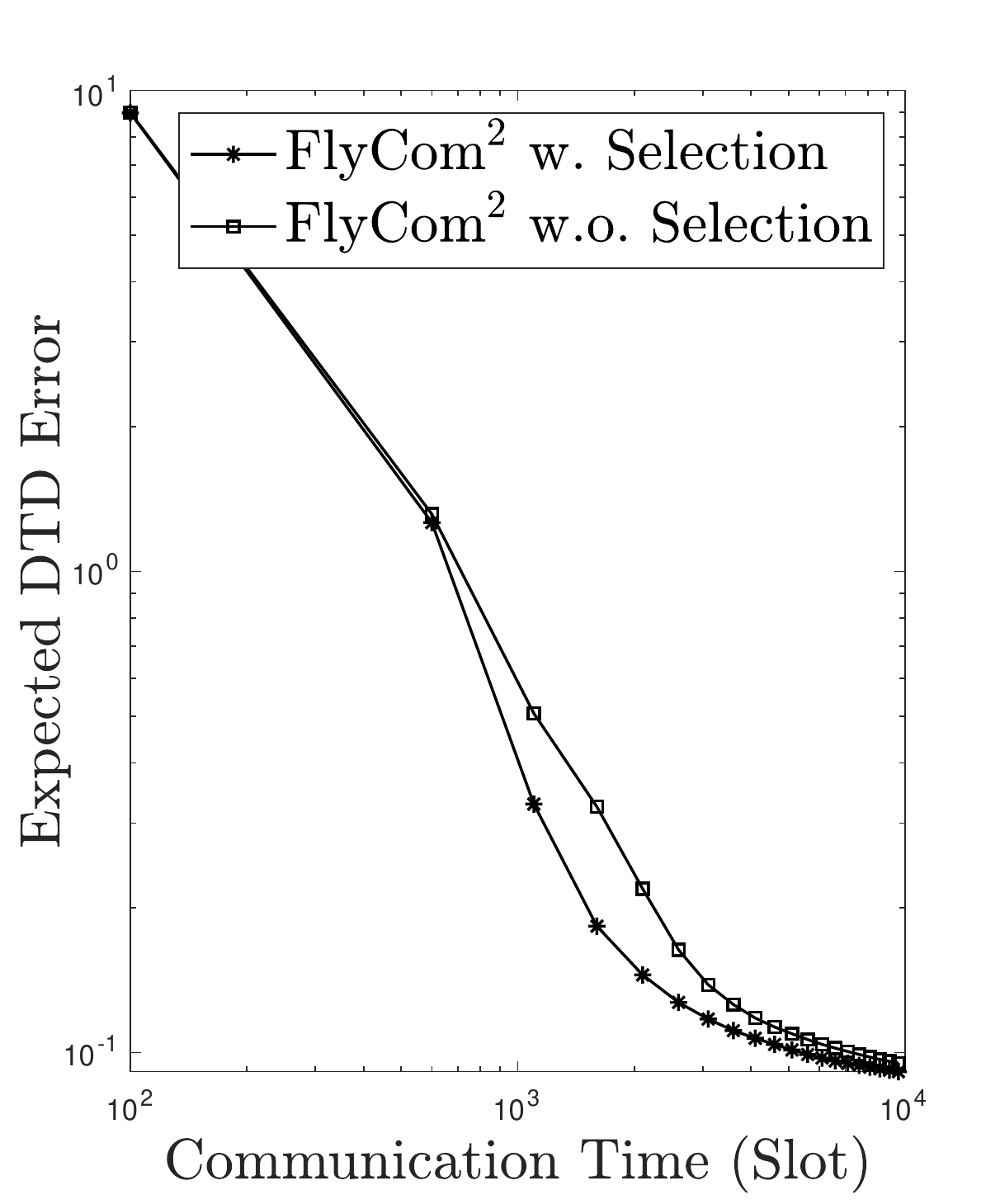}}
    \subfigure[$M= 3$]{\label{subfig:M3}\includegraphics[width=0.24\textwidth]{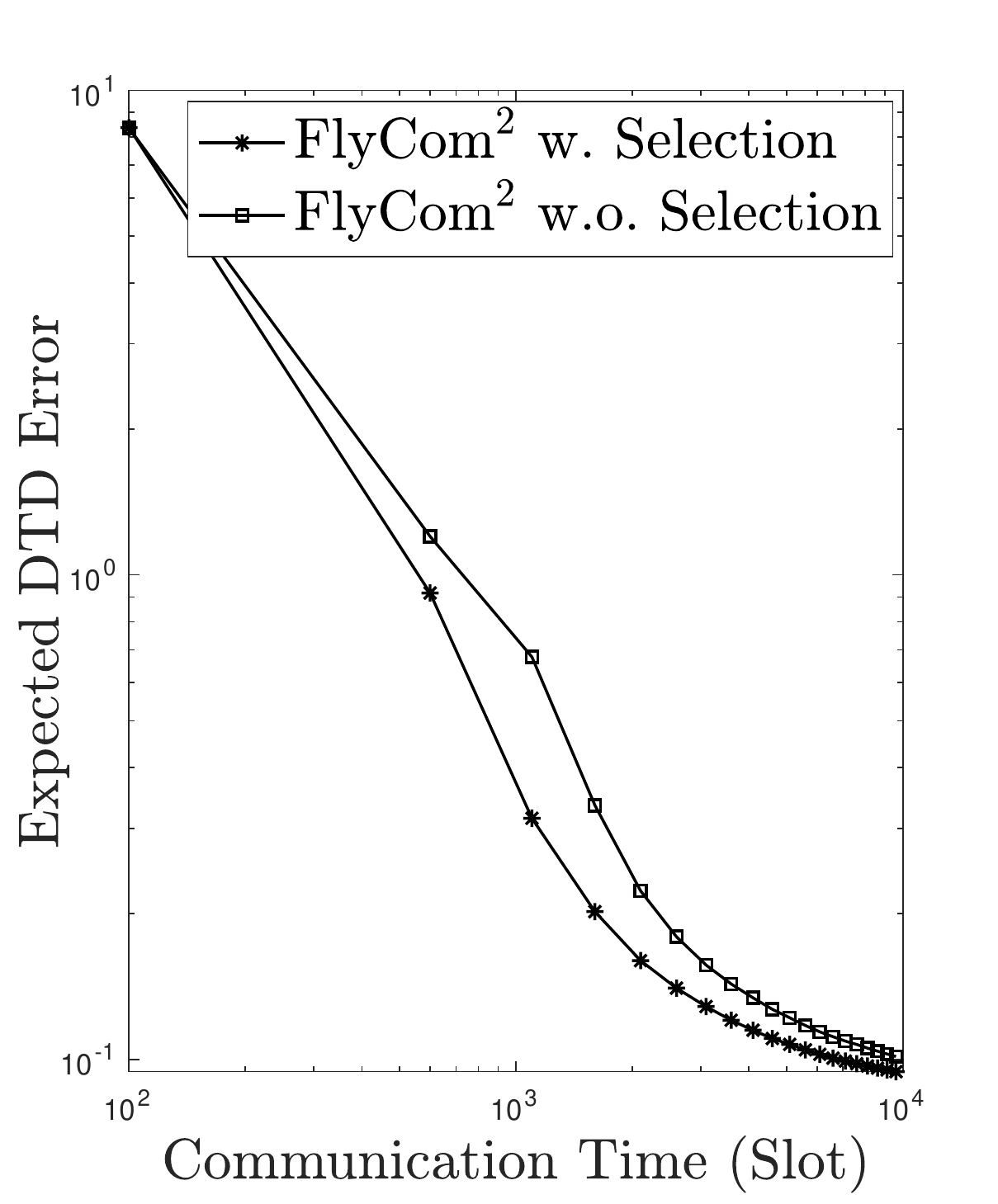}}
    \subfigure[$M= 4$]{\label{subfig:M4}\includegraphics[width=0.24\textwidth]{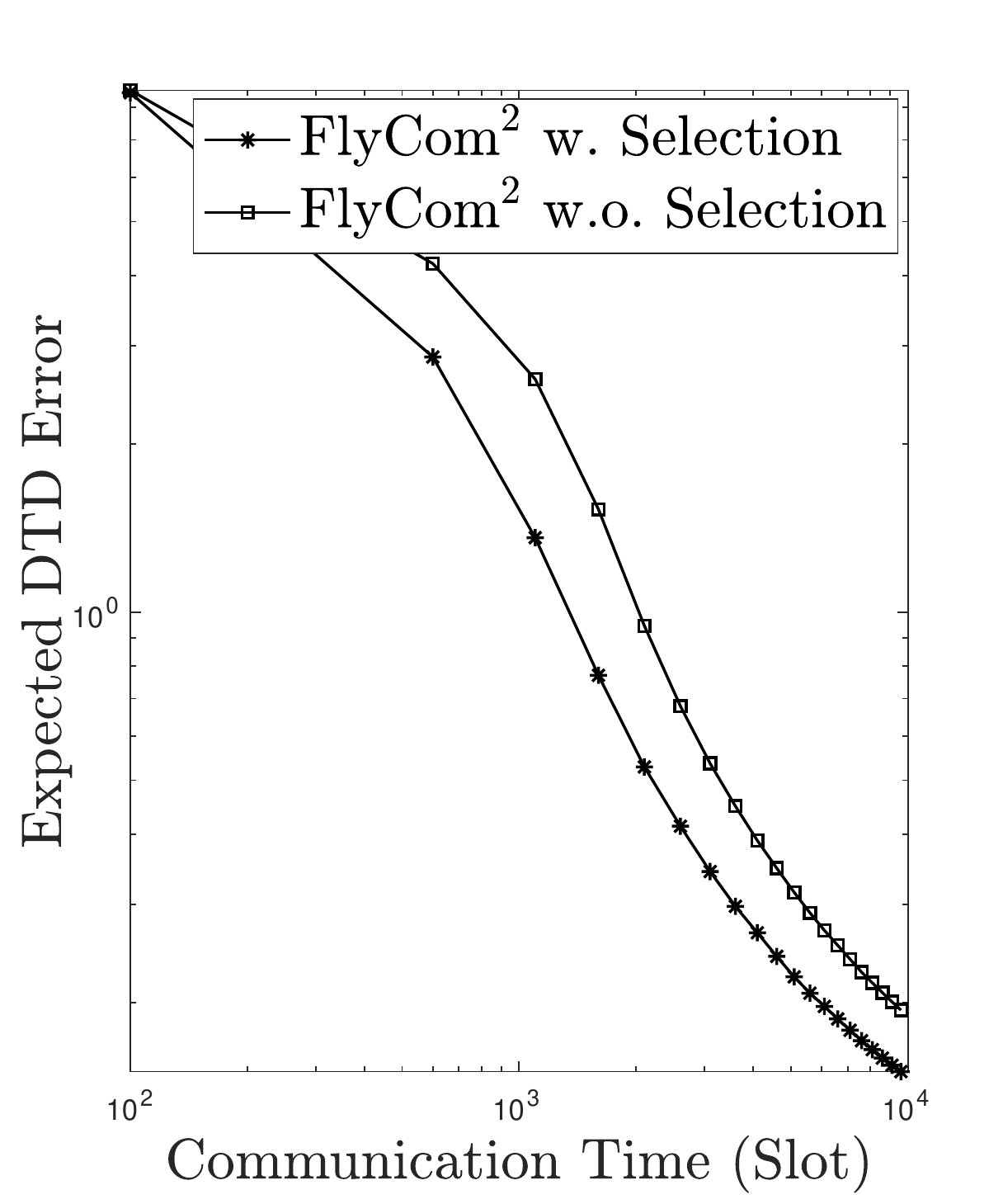}}
    \caption{Error-performance comparison between FlyCom$^2$ with and without sketch selection, SNR $\gamma = 10\mathrm{dB}$.}
    \label{fig:selection}
\end{figure*}
\subsection{Experimental Settings}
First, the MIMO AirComp system is configured to have the following settings. There are $K=20$ edge devices connected to the server. The array sizes at each device and the server are set as $N_{\text{t}}=4$ and $N_{\text{r}}=16$, respectively.  The Rayleigh channel with shadow fading is adopted, in which each MIMO channel is given as $\mathbf{H}_{t,k} = \sqrt{\beta_{t,k}}\hat{\mathbf{H}}_{t,k}$ with $\beta_{t,k}$ following a Gamma distribution $\Gamma(1.2,0.83)$ (see, e.g.~\cite{Shadow_fading2016}) and $\hat{\mathbf{H}}_{t,k}$ comprising i.i.d. $\mathcal{CN}(0,1)$ entries. Different channels are independent. Second, we use a synthetic data model following the DTD literature (see, e.g.~\cite{StreamingTD}). Considering the computation of the $n$-th factor matrix, the unfolding matrix of the data tensor has the size of $I_n\times (\prod_{j=1,j\neq n}^N) = 100\times 1500$, and its columns are uniformly distributed over devices. Under such settings, each local sketch has the length of $I_n = 100$ that is smaller than a single channel coherence block (see, e.g.~\cite{Bjornson2016Tenmyths}, for justification). To demonstrate the performance of the proposed FlyCom$^2$ for data with a range of parameterized spectral distribution, the singular values of the unfolding matrix are set to decay with a polynomial rate:
\begin{equation*}
    \mathbf{\Sigma}_{\mathbf{X}} = \mathsf{diag}\left(1,\cdots,1,\frac{1}{2^{\xi}},\frac{1}{3^{\xi}},\cdots,\frac{1}{(I_n-r)^{\xi}}\right),
\end{equation*}
where the first $r=12$ principal singular values are fixed as $1$ and $\xi>0$ controls the decay rate of residual values. Furthermore, the left and right eigenspaces of the unfolding matrix are generated as those of random matrices with i.i.d. $\mathcal{N}(0,1)$ entries~\cite{StreamingTD}. 

Third, we consider two benchmarking schemes that are variants of SVD-based DTD.
\begin{itemize}
    \item Centroid SVD-DTD: Devices compute local principal eigenspaces $\{\hat{\mathbf{U}}_k\}$ of their on-device data samples by using SVD and the server then aggregates these local results as $\mathbf{P}=\frac{1}{K}\sum_k\hat{\mathbf{U}}_k\hat{\mathbf{U}}_k^{\top}$. The principal eigenspace of $\mathbf{P}$ represents the centroid of all local estimates $\{\hat{\mathbf{U}}_k\}$ on the Grassmannian manifold and is extracted to form a global estimate of the ground truth~\cite{JF2019estimation2019,chen2022analog}.
    \item Alignment SVD-DTD: The scheme follows a similar procedure as above except for aggregating local results, $\{\hat{\mathbf{U}}_k\}$, as $\mathbf{P}=\frac{1}{K}\sum_k\hat{\mathbf{U}}_k\mathbf{J}_k$, where the orthogonal matrices $\{\mathbf{J}_k\}$ are alignment matrices that are to be optimized by using past global estimates to improve the system performance~\cite{VC2021DPCA}. 
\end{itemize}
The aggregation operations in both benchmark schemes are implemented using MIMO AirComp~\cite{GZhu2021WCM,GXZhuAirComp2019} as FlyCom$^2$ for fair comparison.

\subsection{Performance Gain of Sketch Selection}

In Fig.~\ref{fig:selection}, we compare the error performance of FlyCom$^2$ between the cases with and without sketch selection. The communication time is measured by the total number of symbol slots used in uploading local sketches, namely $tI_n$, where $I_n$ is the number of rows of local sketches. We vary the dimension of the receive beamformer, $M$, from $1$ to $4$ to achieve different tradeoffs between channel diversity and multiplexing. It is observed from Fig.~\ref{fig:selection} that the proposed selection scheme helps reduce the expected DTD error for different $M$. The gain emerges when the communication time exceeds $M$-dependent threshold (e.g. $6000$ time slots  for $M=1$) as the total number of available sketches becomes sufficiently large. Another observation from Fig.~\ref{fig:selection} is that the DTD error without selection decreases at an approximately linear rate w.r.t. the communication time, which is aligned with our conclusion in~\eqref{eq:scalinglaw}. In the sequel, FlyCom$^2$ is assumed to have sketch selection with $M$ fixed as $2$.

\begin{figure*}[t]
    \centering
    \subfigure[Spectral decay rate $\xi = 1$]{\label{subfig:decay1}\includegraphics[width=0.3\textwidth]{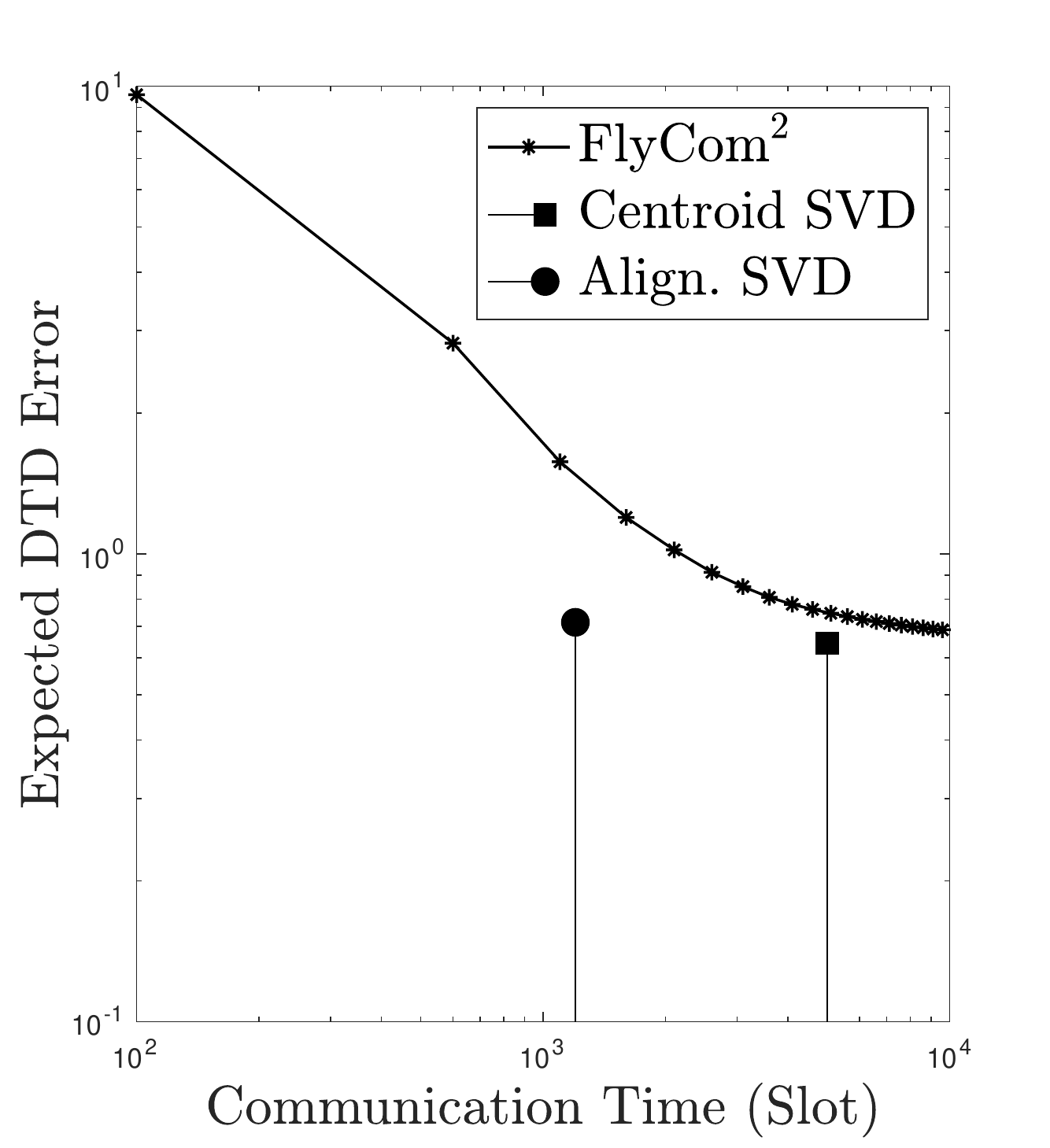}}
    \subfigure[Spectral decay rate $\xi = 2$]{\label{subfig:decay2}\includegraphics[width=0.3\textwidth]{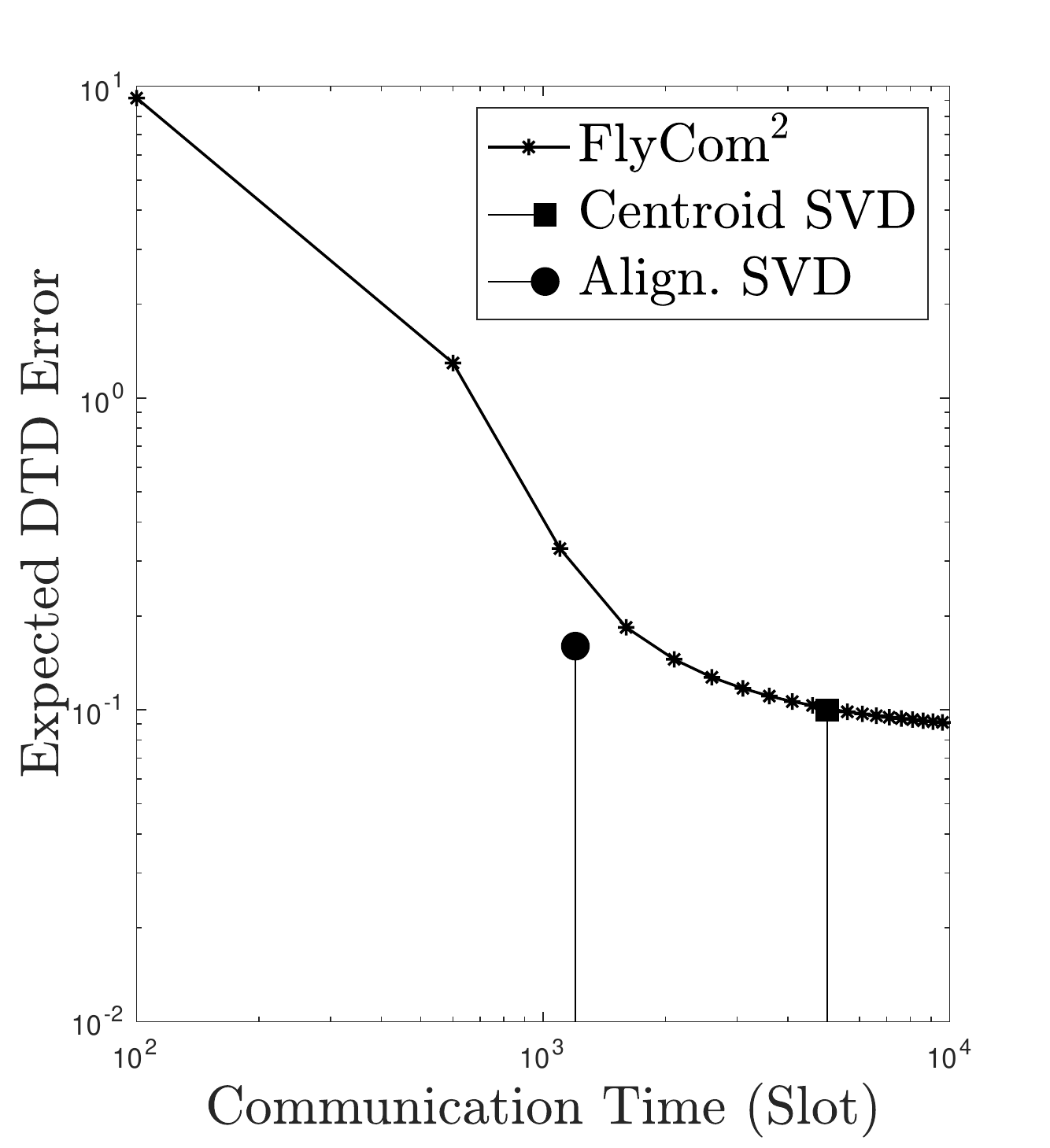}}
    \subfigure[Spectral decay rate $\xi = 3$]{\label{subfig:decay3}\includegraphics[width=0.3\textwidth]{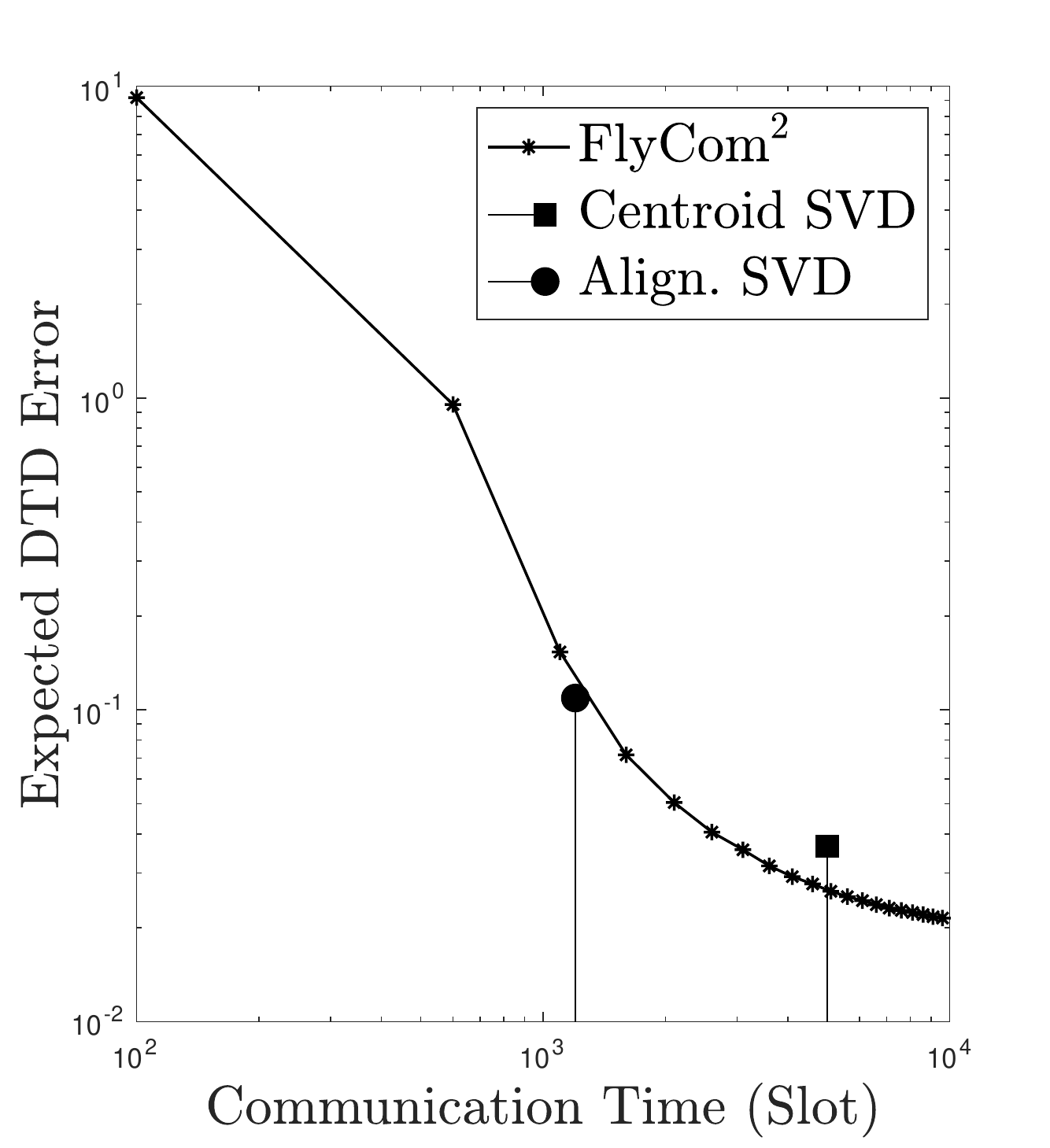}}
    \caption{FlyCom$^2$ versus Benchmark schemes, SNR $\gamma = 10\mathrm{dB}$.}
    \label{fig:benchmark}
\end{figure*}

\begin{figure*}[t]
    \centering
    \subfigure[Computation complexity ($I=100$, $M=4$)]{\label{subfig:complexity}\includegraphics[width=0.44\textwidth]{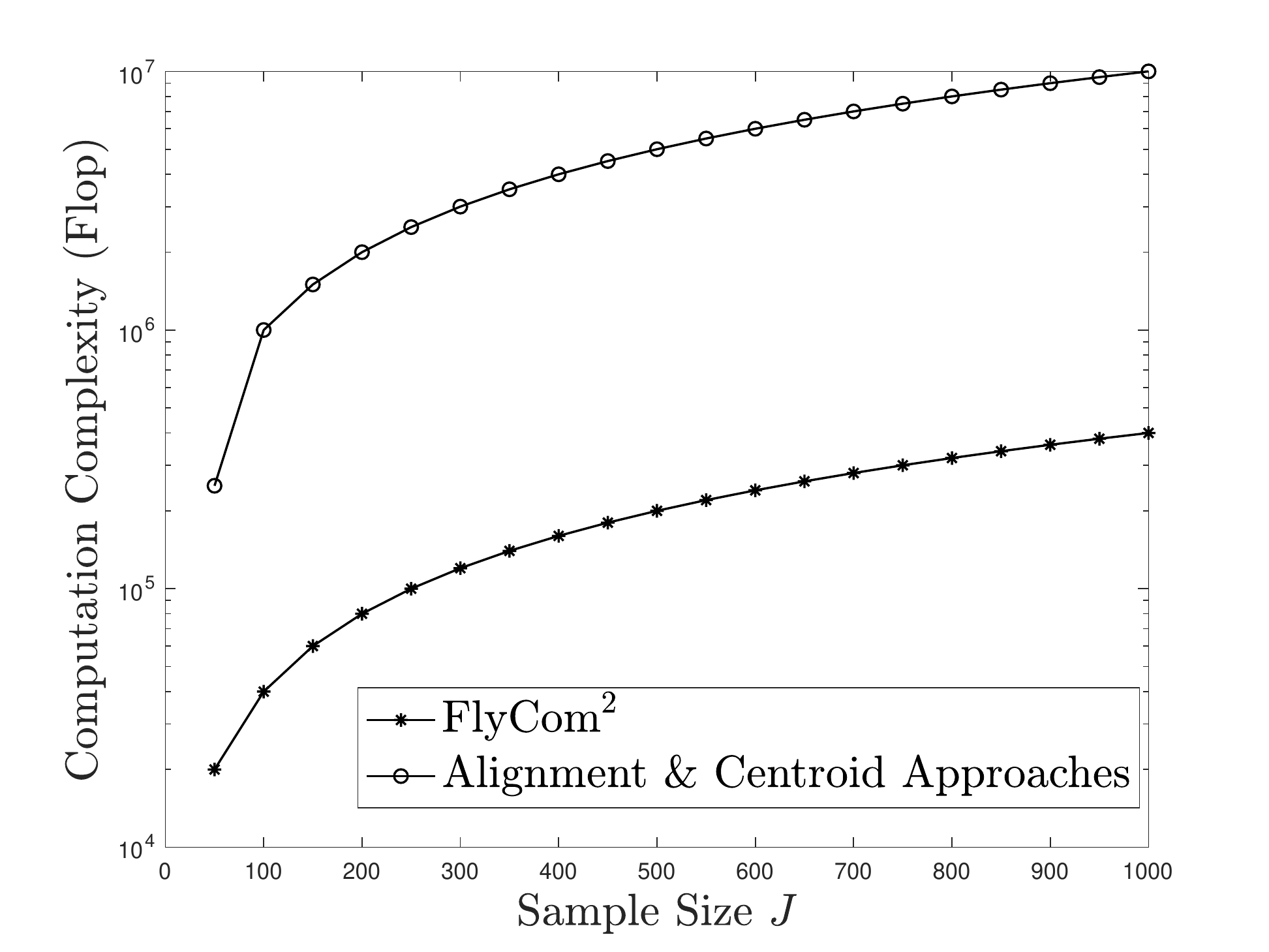}}
    \subfigure[Passes of raw data in memory]{\label{subfig:memory}\includegraphics[width=0.44\textwidth]{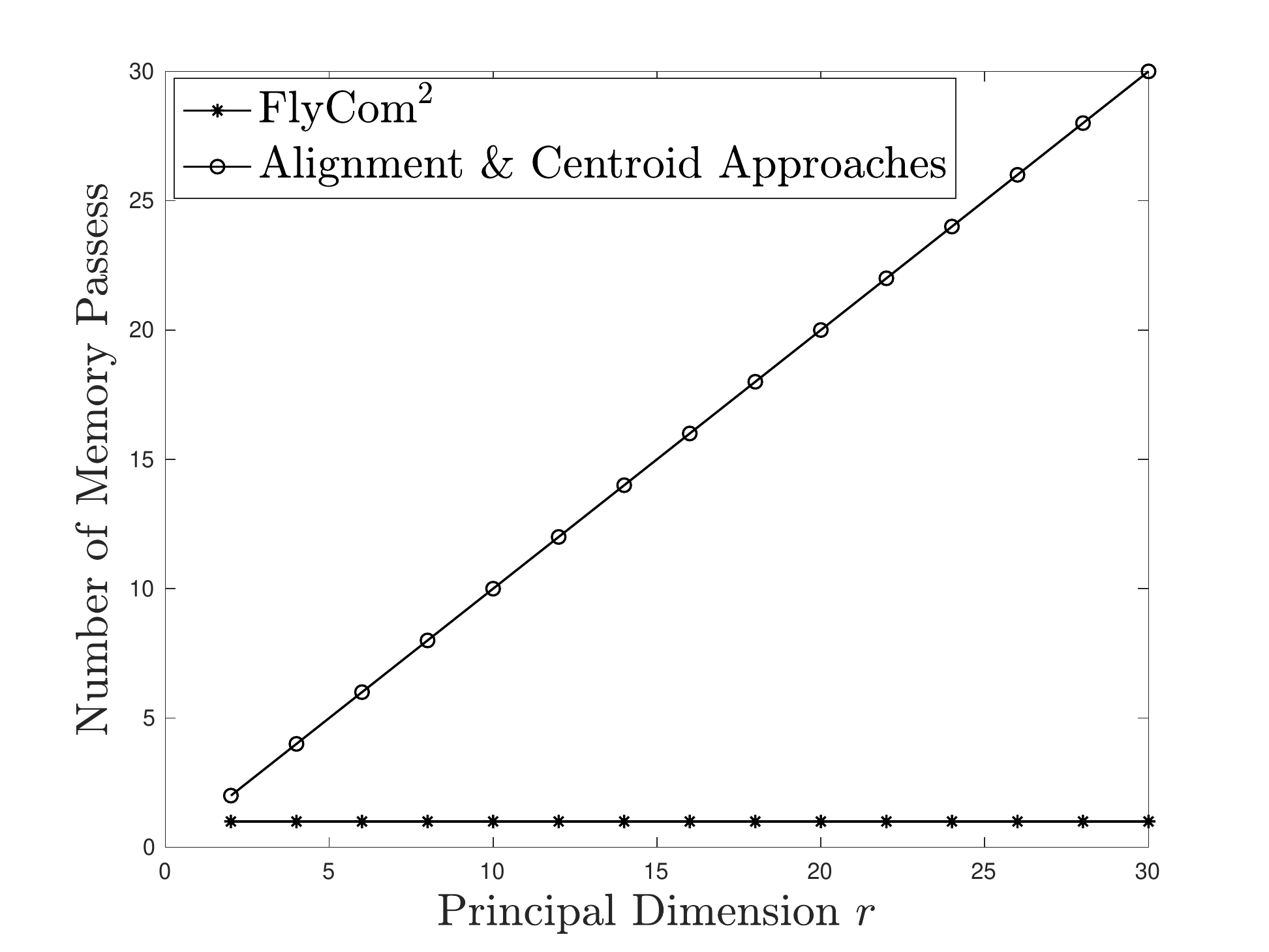}}
    \caption{Computation-cost comparison between FlyCom$^2$ based DTD and benchmarking schemes.}
    \label{fig:computation}
\end{figure*}
\subsection{Error Performance of FlyCom$^2$}

While FlyCom$^2$ requires much simpler on-device computation than benchmarking schemes (see Section~\ref{subsec:experiment_error}), we demonstrate in Fig.~\ref{fig:benchmark} that it can achieve comparable or even better error performance than the latter. Fig.~\ref{fig:benchmark} shows the expected DTD error versus the communication time. The performance of the benchmark schemes with one-shot computation and communication appears as single points in the figure. The results in Fig.~\ref{fig:benchmark} show that FlyCom$^2$-based DTD achieves comparable decomposition accuracy as the benchmark schemes with  progressing time. Furthermore, its performance is improved by increasing the decay rate ($\xi$) of the singular values, which validates our conclusion that large eigen-gaps help distinguish principal from non-principal eigenvectors during random sketching. For instance, for $\xi=2$, the proposed scheme approaches the centroid and alignment based SVD-DTD in performance for communication time larger than $2600$ and $4000$ symbol slots, respectively. As $\xi$ increases to $3$, the former achieves the same error performance as the alignment-based method while outperforming the centroid based method. Furthermore, one can observe from Fig.~\ref{fig:benchmark} that the proposed on-the-fly framework realizes a flexible trade-off between the decomposition accuracy and communication time, which is the distinctive feature of the design.

\subsection{Device Computation Costs of FlyCom$^2$}\label{subsec:experiment_error}
In Fig.~\ref{fig:computation}, we compare two kinds of computational costs at devices, namely complexity and memory passes, between FlyCom$^2$ and benchmark schemes. The complexity refers to the flop count of computation, and the memory passes are equal to the number of memory visits for reading data entries. The computational advantage of FlyCom$^2$ is demonstrated by comparing the cost of matrix-vector multiplication in random sketching with that of deterministic SVD used in the one-shot benchmarking schemes. Specifically, given $I\times J$ local unfolding matrices, deterministic SVD has the complexity proportional to $\min\{I,J\}^2\times \max\{I,J\}$~\cite{SVDComplexity2019}; based on matrix multiplication, the complexity of FlyCom$^2$ to yield an $M$-dimensional sketch at each time slot is $IJM$. For the schemes in comparison, their curves of computation complexity versus sample size are plotted in Fig.~\ref{subfig:complexity}. One can observe that the proposed FlyCom$^2$ dramatically reduces devices' complexity by more than an order of magnitude. On the other hand, Fig.~\ref{subfig:memory} displays the curves of the number of memory passes versus the principal dimensionality, $r$. The proposed design keeps a constant memory pass for matrix multiplication, as opposed to that of SVD which increases linearly with the principal dimensionality. For example, the number of memory passes is reduced using FlyCom$^2$ by $30$ times for $r=30$.

%
\section{Conclusion}\label{section:conclusion}
We have proposed the FlyCom$^2$ framework, that supports the progressive computation of DTD in mobile networks. Through the use of  random sketching techniques at devices, the traditional one-shot high-dimensional mobile communication and computation is reduced to low-dimensional operations spread over multiple time slots. Thereby, the resource constraints of devices are overcome. Furthermore, FlyCom$^2$ obtains its distinctive feature of progressive improvement of DTD accuracy with increasing communication time, providing robustness against link disruptions. To develop the FlyCom$^2$ based DTD framework, we have designed an on-the-fly sub-space estimator and a sketch-selection scheme to ensure close-to-optimal system performance.

Beyond DTD, high-dimensional communication and computation pose a general challenge for machine learning and data analytics in wireless networks. We expect that FlyCom$^2$ can be further developed into a broad approach for efficient deployment of relevant algorithms such as federated learning and distributed optimization. For the current FlyCom$^2$ targeting DTD, its extension to accommodate other wireless techniques such as broadband transmission and radio resource management is also a direction worth pursuing.

\appendix
\subsection{Proof of Lemma~\ref{Lemma:GaussianMatrix}}\label{Apdx:representation}
In~\eqref{eq:observations}, both the $\mathbf{\mathbf{F}}_{t}$ and $\tilde{\mathbf{Z}}_t$ have i.i.d.  zero-mean Gaussian entries, thereby enforcing the observation $\tilde{\mathbf{Y}}_t$ to be a Gaussian matrix. This conclusion holds for all observations and thus the aggregation $\hat{\mathbf{Y}}_t$ is a Gaussian matrix and can be decomposed into $\mathbf{C}^{\frac{1}{2}}\mathbf{W}\mathbf{D}^{\frac{1}{2}}$. Therein, $\mathbf{W}$ has i.i.d. $\mathcal{N}(0,1)$ entries. The covariance matrices $\mathbf{C}$ and $\mathbf{D}$ are computed as follows. First, there is $\mathsf{E}[\hat{\mathbf{Y}}_t\hat{\mathbf{Y}}_t^{\top}]=\mathsf{E}[\mathbf{C}^{\frac{1}{2}}\mathbf{W}\mathbf{D}\mathbf{W}^{\top}\mathbf{C}^{\frac{1}{2}}]$, where the right hand side of the equation equals to $\mathbf{C}\mathsf{Tr}(\mathbf{D})$ while the left hand side is given as
\begin{align*}
    \mathsf{E}[\hat{\mathbf{Y}}_t\hat{\mathbf{Y}}_t^{\top}]&=\sum_{\ell\leq t}\mathsf{E}[\tilde{\mathbf{Y}}_t\tilde{\mathbf{Y}}_t^{\top}]\\
    & = tM\mathbf{X}\mathbf{X}^{\top} + \frac{1}{2}\sigma^2\sum_{\ell\leq t}\mathsf{Tr}(\mathbf{A}_{\ell}^{H}\mathbf{A}_{\ell})\mathbf{I}_I.
\end{align*}
On the other hand, using $\mathsf{E}[\hat{\mathbf{Y}}_t^{\top}\hat{\mathbf{Y}}_t] = \mathsf{E}[\mathbf{D}^{\frac{1}{2}}\mathbf{W}^{\top}\mathbf{C}\mathbf{W}\mathbf{D}^{\frac{1}{2}}] = \mathbf{D}\mathsf{Tr}(\mathbf{C})$, we have
\begin{align*}
    \mathbf{D}\mathsf{Tr}(\mathbf{C})&=\mathsf{E}\left[[\tilde{\mathbf{Y}}_1,\tilde{\mathbf{Y}}_2,\cdots,\tilde{\mathbf{Y}}_t]^{\top}[\tilde{\mathbf{Y}}_1,\tilde{\mathbf{Y}}_2,\cdots,\tilde{\mathbf{Y}}_t]\right]\\
    & = \mathsf{diag}\left(\mathsf{E}[\tilde{\mathbf{Y}}_1^{\top}\tilde{\mathbf{Y}}_1],\cdots,\mathsf{E}[\tilde{\mathbf{Y}}_t^{\top}\tilde{\mathbf{Y}}_t]\right),
\end{align*}
where an arbitrary diagonal block, say $\mathsf{E}[\tilde{\mathbf{Y}}_{\ell}^{\top}\tilde{\mathbf{Y}}_{\ell}]$, $\forall \ell\leq t$, can be expressed as
\begin{align*}
    \mathsf{E}[\tilde{\mathbf{Y}}_{\ell}^{\top}\tilde{\mathbf{Y}}_{\ell}] & = \mathsf{E}[\mathbf{F}^{\top}\mathbf{X}^{\top}\mathbf{X}\mathbf{F}] + \mathsf{E}[\tilde{\mathbf{Z}}_t\tilde{\mathbf{Z}}_t^{\top}]\\
    & = \mathsf{Tr}(\mathbf{X}^{\top}\mathbf{X})\mathbf{I}_{M} + \frac{1}{2}I\sigma^2\mathbf{A}_{\ell}\mathbf{A}_{\ell}^{H}.
\end{align*} 
Concluding the above results yields the covariance matrices as $\mathbf{C} = \mathbf{X}\mathbf{X}^{\top} + \frac{1}{2tM}\sigma^2\sum_{\ell\leq t}\mathsf{Tr}(\mathbf{A}_{\ell}^{H}\mathbf{A}_{\ell})\mathbf{I}_I$ and $\mathbf{D} = 
    \frac{tM\mathsf{Tr}(\mathbf{X}^{\top}\mathbf{X})\mathbf{I}_{tM} + \frac{1}{2}ItM\sigma^2\mathsf{diag}(\mathbf{A}_1\mathbf{A}_1^{H},\cdots,\mathbf{A}_t\mathbf{A}_t^{H})}{tM\mathsf{Tr}(\mathbf{X}^{\top}\mathbf{X}) + \frac{1}{2}I\sigma^2\sum_{\ell\leq t}\mathsf{Tr}(\mathbf{A}_{\ell}^{H}\mathbf{A}_{\ell})}$, respectively. This completes the proof.

\subsection{Proof of Lemma~\ref{Lemma:deviation}}\label{Apdx:deviation}
First, rewrite the error as
\begin{align*}
    &\Vert(\mathbf{I}_I-\tilde{\mathbf{U}}\tilde{\mathbf{U}}^{\top})\mathbf{X}\Vert_F^2\\
    & = \sum_{j=1}^I\sigma_j^2 -\sum_{i=1}^r\sum_{j\neq i}\sigma_j^2\langle\tilde{\mathbf{u}}_i,\mathbf{u}_j\rangle^2 - \sum_{i=1}^r\sigma_i^2\langle\tilde{\mathbf{u}}_i,\mathbf{u}_i\rangle^2\\
    & =  \sum_{j\geq r+1}\sigma_j^2 - \sum_{i=1}^r\sum_{j\neq i}(\sigma_i^2-\sigma_j^2)\langle\tilde{\mathbf{u}}_i,\mathbf{u}_j\rangle^2,
\end{align*}
where the last step is due to $\langle\tilde{\mathbf{u}}_i,\mathbf{u}_i\rangle^2 = 1 - \sum_{j\neq i}\langle\tilde{\mathbf{u}}_i,\mathbf{u}_j\rangle^2$. 

Then, under the assumption of $\sigma_i=\sigma_j$, $\forall i,j\leq r$, the second term on the right side of the above equation can be rewritten as
\begin{align*}
    &\sum_{i=1}^r\sum_{j\neq i}(\sigma_i^2-\sigma_j^2)\langle\tilde{\mathbf{u}}_i,\mathbf{u}_j\rangle^2\\
    & = \sum_{i=1}^r\sum_{j\geq r+1}^I(\sigma_i^2-\sigma_j^2)\langle\tilde{\mathbf{u}}_i,\mathbf{u}_j\rangle^2 \\
    &\quad +\underbrace{\sum_{i=1}^r\sum_{j=1,j\neq i}^r(\sigma_i^2-\sigma_j^2)\langle\tilde{\mathbf{u}}_i,\mathbf{u}_j\rangle^2}_{=0}\\
    & = \sum_{i=1}^r\sum_{j\geq r+1}^I(\sigma_i^2-\sigma_j^2)\langle\tilde{\mathbf{u}}_i,\mathbf{u}_j\rangle^2.
\end{align*}
Putting the results above together, the conclusion in Lemma~\ref{Lemma:deviation} follows.

\subsection{Proof of Lemma~\ref{Lemma:UPofEachTerm}}\label{Apdx:deterministicerror}
According to Remark~\ref{Remark:analogtransmission}, we have $\lambda_i>\lambda_j$, $\forall i\leq r<j$. Then, using the perturbation theory~\cite[Theorem 3.1 \& Theorem 3.2]{pmlr_loukas17a}, the $\langle\tilde{\mathbf{u}}_i,\mathbf{u}_j\rangle$ can be upper bounded as
\begin{equation*}
 \langle\tilde{\mathbf{u}}_i,\mathbf{u}_j\rangle\leq \min\left\{\max \left\{2,2\frac{|\tilde{\lambda}_i-\lambda_i|}{|\tilde{\lambda}_i-\lambda_j|}\right\},\frac{\lambda_i-\lambda_j}{|\tilde{\lambda}_i-\lambda_j|}\right\}\frac{\Vert\mathbf{\Delta}\mathbf{u}_j\Vert_2}{\lambda_i-\lambda_j},
\end{equation*}
where $\lambda_i$ and $\tilde{\lambda}_i$ denote the $i$-th eigenvalues of $\mathbf{\Lambda}$ and $\frac{1}{tM}\mathbf{\Phi}_t\mathbf{\Phi}_t^{\top}$, respectively. Based on $\min\{\max\{a,b\},c\}= \max\{\min\{a,c\},\min\{b,c\}\}\leq \max\{a,\min\{b,c\}\}$, the upper bound can be further written as
\begin{align*}
\langle\tilde{\mathbf{u}}_i,\mathbf{u}_j\rangle&\leq \max \left\{2,\min\left\{\frac{2|\tilde{\lambda}_i-\lambda_i|}{|\tilde{\lambda}_i-\lambda_j|},\frac{\lambda_i-\lambda_j}{|\tilde{\lambda}_i-\lambda_j|}\right\}\right\}\frac{\Vert\mathbf{\Delta}\mathbf{u}_j\Vert_2}{\lambda_i-\lambda_j}\\
& =  \max \left\{2,\frac{2\min\{|\tilde{\lambda}_i-\lambda_i|,\lambda_i-\lambda_j\}}{|\tilde{\lambda}_i-\lambda_j|}\right\}\frac{\Vert\mathbf{\Delta}\mathbf{u}_j\Vert_2}{\lambda_i-\lambda_j}.
\end{align*}
Recall that $\lambda_i = \sigma_i^2 + \sigma^2\sum_{\ell\leq t}\mathsf{Tr}(\mathbf{A}_{\ell}^{H}\mathbf{A}_{\ell})/2tM$, we have $\lambda_i-\lambda_j = \sigma_i^2-\sigma_j^2$, which yields
\begin{align*}
    &(\sigma_i^2-\sigma_j^2)\langle\tilde{\mathbf{u}}_i,\mathbf{u}_j\rangle^2\\
    & = \max \left\{4,\frac{\min\{4|\tilde{\lambda}_i-\lambda_i|^2,(\sigma_i^2-\sigma_j^2)^2\}}{|\tilde{\lambda}_i-\lambda_j|^2}\right\}\frac{\Vert\mathbf{\Delta}\mathbf{u}_j\Vert_2^2}{\sigma_i^2-\sigma_j^2},
\end{align*}
which completes the proof.

\subsection{Proof of Theorem~\ref{Theorem:expectation}}\label{Apdx:expectation}
First, the square vector norm, $\Vert\mathbf{\Delta}\mathbf{u}_j\Vert_2^2$, can be rewritten as 
\begin{align*}
\Vert\mathbf{\Delta}\mathbf{u}_j\Vert_2^2 = &\mathbf{u}_j^{\top}\left(\frac{1}{tM}\mathbf{\Phi}_t\mathbf{\Phi}_t^{\top}-\mathbf{U}_{\mathbf{X}}\mathbf{\Lambda}\mathbf{U}_{\mathbf{X}}^{\top}\right)^2\mathbf{u}_j\\
 = &\lambda_j^2-\frac{2\lambda_j^2}{tM}\mathbf{e}_j\mathbf{W}\mathbf{W}^{\top}\mathbf{e}_j^{\top}\\
 &+\frac{\lambda_j}{(tM)^2}\mathbf{e}_j\mathbf{W}\mathbf{W}^{\top}\mathbf{\Lambda}\mathbf{W}\mathbf{W}^{\top}\mathbf{e}_j^{\top},
\end{align*}
where we define $\mathbf{e}_j = [0,\cdots,0,1,0,\cdots,0]$ with the $j$-th element being $1$ and other elements being $0$. Since $\mathbf{W}$ has i.i.d. $\mathcal{N}(0,1)$ entries, the expectation of the upper bound of $(\sigma_i^2-\sigma_j^2)\langle\tilde{\mathbf{u}}_i,\mathbf{u}_j\rangle^2$ can be expressed as
\begin{align*}
    &(\sigma_i^2-\sigma_j^2)\mathsf{E}[\langle\tilde{\mathbf{u}}_i,\mathbf{u}_j\rangle^2]\\
    & \leq\frac{4}{\sigma_i^2-\sigma_j^2}\mathsf{E}[\Vert\mathbf{\Delta}\mathbf{u}_j\Vert_2^2]\\
    & = \frac{4}{\sigma_i^2-\sigma_j^2}\left[\frac{\lambda_j}{(tM)^2}\mathbf{e}_j\mathsf{E}[\mathbf{W}\mathbf{W}^{\top}\mathbf{\Lambda}\mathbf{W}\mathbf{W}^{\top}]\mathbf{e}_j^{\top}-\lambda_j^2\right].
\end{align*}
The result of $\mathbf{e}_j\mathsf{E}[\mathbf{W}\mathbf{W}^{\top}\mathbf{\Lambda}\mathbf{W}\mathbf{W}^{\top}]\mathbf{e}_j^{\top}$ can be derived by representing $\mathbf{W} = [\mathbf{w}_1^{\top},\mathbf{w}_2^{\top},\cdots,\mathbf{w}_I^{\top}]^{\top}$, where $\mathbf{w}_{i}$ has i.i.d. $\mathcal{N}(0,1)$ entries and is independent with $\mathbf{w}_{j}$ with $i\neq j$. In specific, there is
\begin{align*}
    &\mathbf{e}_j\mathsf{E}[\mathbf{W}\mathbf{W}^{\top}\mathbf{\Lambda}\mathbf{W}\mathbf{W}^{\top}]\mathbf{e}_j^{\top}\\
    & = \lambda_j\mathsf{E}[\mathbf{w}_j\mathbf{w}_j^{\top}\mathbf{w}_j\mathbf{w}_j^{\top}] + \sum_{i\neq j}\lambda_i\mathsf{E}[\mathbf{w}_j\mathsf{E}[\mathbf{w}_i^{\top}\mathbf{w}_i]\mathbf{w}_j^{\top}]\\
    & = tM[\mathsf{Tr}(\mathbf{\Lambda}) + (tM +1)\lambda_j].
\end{align*}
Putting the above results together yields 
\begin{equation*}
    (\sigma_i^2-\sigma_j^2)\mathsf{E}[\langle\tilde{\mathbf{u}}_i,\mathbf{u}_j\rangle^2]\leq \frac{4[\lambda_j^2 + \lambda_j\mathsf{Tr}(\mathbf{\Lambda})]}{tM(\sigma_i^2-\sigma_j^2)},
\end{equation*}
which completes the proof.
\subsection{Proof of Theorem~\ref{Theorem:probability}}\label{Apdx:probability}
Using Lemma~\ref{Lemma:UPofEachTerm}, we rewrite the error bound as 
\begin{align*}
    &\Vert(\mathbf{I}_I-\tilde{\mathbf{U}}\tilde{\mathbf{U}}^{\top})\mathbf{X}\Vert_F^2-\sum_{i\geq r+1}\sigma_i^2\\
    &\leq g(\mathbf{W})+ \sum_{i=1}^r\sum_{j\geq r+1}\frac{4\lambda_j^2}{\sigma_i^2-\sigma_j^2},
\end{align*}
where the involved random part is defined as $g(\mathbf{W})\overset{\triangle}{=}\mathsf{Tr}\left(\mathbf{W}\mathbf{W}^{\top}\mathbf{\Lambda}\mathbf{W}\mathbf{W}^{\top}\mathbf{\Lambda}_1\right)- \mathsf{Tr}\left(\mathbf{W}\mathbf{W}^{\top}\mathbf{\Lambda}_2\right)$ with $\mathbf{\Lambda}_1 = \sum_{i=1}^r\sum_{j\geq r+1}\frac{4}{\sigma_i^2-\sigma_j^2}\frac{\lambda_j}{(tM)^2}\mathbf{e}_j^{\top}\mathbf{e}_j$
and $
    \mathbf{\Lambda}_2 = \sum_{i=1}^r\sum_{j\geq r+1}\frac{8\lambda_j^2}{(\sigma_i^2-\sigma_j^2)tM}\mathbf{e}_j^{\top}\mathbf{e}_j$. It then follows that the probabilistic error bound is upper bounded as 
\begin{align*}
    &\mathsf{Pr}\left[\Vert(\mathbf{I}_I-\tilde{\mathbf{U}}\tilde{\mathbf{U}}^{\top})\mathbf{X}\Vert_F^2-\sum_{i\geq r+1}\sigma_i^2\geq \mathsf{E}[g(\mathbf{W})] + \epsilon\right]\\
    &\leq \mathsf{Pr}\left[g(\mathbf{W})-\mathsf{E}[g(\mathbf{W})]\geq \epsilon\right].
\end{align*}
Next, rewrite $\mathbf{W} = [\tilde{\mathbf{w}}_1,\tilde{\mathbf{w}}_2,\cdots,\tilde{\mathbf{w}}_{tM}]$ and the random variable $g(\mathbf{W})$ can be rewritten as
\begin{equation*}
    g(\mathbf{W}) =\sum_{m_1,m_2}\tilde{\mathbf{w}}_{m_1}^{\top}\mathbf{\Lambda}\tilde{\mathbf{w}}_{m_2}\tilde{\mathbf{w}}_{m_2}^{\top}\mathbf{\Lambda}_1\tilde{\mathbf{w}}_{m_1} - \sum_{m_3}\tilde{\mathbf{w}}_{m_3}^{\top}\mathbf{\Lambda}_2\tilde{\mathbf{w}}_{m_3}.
\end{equation*}
Let $\hat{\mathbf{w}}_m$ be an independent copy of $\tilde{\mathbf{w}}_m$ and then there is
\begin{align*}
    &|g^{\prime}\left(\{\tilde{\mathbf{w}}_{M}\}_{M\neq m},\tilde{\mathbf{w}}_m\right) - g^{\prime}\left(\{\tilde{\mathbf{w}}_{M}\}_{M\neq m},\hat{\mathbf{w}}_m\right)|\\
    &=|\hat{\mathbf{w}}_m^{\top}\mathbf{\Lambda}_2\hat{\mathbf{w}}_m-\tilde{\mathbf{w}}_m^{\top}\mathbf{\Lambda}_2\tilde{\mathbf{w}}_m \\
    &\quad+ \tilde{\mathbf{w}}_m^{\top}\mathbf{\Lambda}\tilde{\mathbf{w}}_m\tilde{\mathbf{w}}_m^{\top}\mathbf{\Lambda}_1\tilde{\mathbf{w}}_m - \hat{\mathbf{w}}_m^{\top}\mathbf{\Lambda}\hat{\mathbf{w}}_m\hat{\mathbf{w}}_m^{\top}\mathbf{\Lambda}_1\hat{\mathbf{w}}_m \\
    &\quad + \sum_{m_1\neq m}\tilde{\mathbf{w}}_{m_1}^{\top}\mathbf{\Lambda}(\tilde{\mathbf{w}}_{m}\tilde{\mathbf{w}}_{m}^{\top} - \hat{\mathbf{w}}_{m}\hat{\mathbf{w}}_{m}^{\top})\mathbf{\Lambda}_1\tilde{\mathbf{w}}_{m_1}  \\
    &\quad+ \sum_{m_2\neq m}\tilde{\mathbf{w}}_{m_2}^{\top}\mathbf{\Lambda}_1(\tilde{\mathbf{w}}_{m}\tilde{\mathbf{w}}_{m}^{\top} - \hat{\mathbf{w}}_{m}\hat{\mathbf{w}}_{m}^{\top})\mathbf{\Lambda}\tilde{\mathbf{w}}_{m_2}|.
\end{align*}
Note that the above equation does not have an upper bound since the Gaussian random variables go from minus infinity to infinity. To endow on $g(\mathbf{W})$ the bounded difference property, the Gaussian concentration is exploited as follows. Specifically, a $\mathcal{N}(0,1)$ variable $w$ can be smaller than a threshold, say $\kappa>1$, with the probability of $p(|w|\leq \kappa) = \mathsf{erf}\left(\frac{\kappa}{\sqrt{2}}\right)\overset{\triangle}{=}p_{\kappa}$, where $\mathsf{erf}\left(\cdot\right)$ denotes the error function. Hence, let the event that the abstract value of the last $I-r$ elements in the vectors $\tilde{\mathbf{w}}_m$ and $\hat{\mathbf{w}}_m$ are bounded by $\kappa$ be denoted by $\mathsf{BD}$ and its complement by $\mathsf{UBD}$. Then, there are $\mathsf{Pr}(\mathsf{BD}) = p_{\kappa}^{tM(I-r)}$ and $\mathsf{Pr}(\mathsf{UBD}) = 1-p_{\kappa}^{tM(I-r)}$. Hence, with the probability of $\mathsf{Pr}(\mathsf{BD})$, $|g^{\prime}\left(\{\tilde{\mathbf{w}}_{M}\}_{M\neq m},\tilde{\mathbf{w}}_m\right) - g^{\prime}\left(\{\tilde{\mathbf{w}}_{M}\}_{M\neq m},\hat{\mathbf{w}}_m\right)|$ can be upper bounded as
\begin{align*}
    &|g^{\prime}\left(\{\tilde{\mathbf{w}}_{M}\}_{M\neq m},\tilde{\mathbf{w}}_m\right) - g^{\prime}\left(\{\tilde{\mathbf{w}}_{M}\}_{M\neq m},\hat{\mathbf{w}}_m\right)| \\
    &\leq \kappa^4\mathsf{Tr}(\mathbf{\Lambda})\mathsf{Tr}(\mathbf{\Lambda}_1) - \kappa^2\mathsf{Tr}(\mathbf{\Lambda}_2) + 2(tM-1)\kappa^4\mathsf{Tr}(\mathbf{\Lambda})\mathsf{Tr}(\mathbf{\Lambda}_1)\\
    & =  \sum_{i=1}^r\sum_{j\geq r+1}\frac{8\lambda_j}{(\sigma_i^2-\sigma_j^2)tM}\left[\frac{tM-1/2}{tM}\kappa^4\mathsf{Tr}(\mathbf{\Lambda})-\lambda_j\kappa^2\right]\\
    &\leq 2\kappa^4\sum_{i=1}^r\sum_{j\geq r+1}\frac{4\lambda_j}{(\sigma_i^2-\sigma_j^2)tM}\left[\mathsf{Tr}(\mathbf{\Lambda})+\lambda_j\right]\\
    & = 2\kappa^4 \mathsf{E}[g(\mathbf{W})],
\end{align*}
which allows us to leverage the concentration theorem shown in Lemma~\ref{Lemma:McDiarmid} to give
\begin{equation*}
    \mathsf{Pr}\left[g(\mathbf{W})-\mathsf{E}[g(\mathbf{W})]\geq \epsilon| \mathsf{BD}\right]\leq \exp\left(-\frac{\epsilon^2}{2\kappa^8 \mathsf{E}^2[g(\mathbf{W})]}\right).
\end{equation*}
Concluding both cases of $\mathsf{BD}$ and $\mathsf{UBD}$, the probabilistic error bound can be expressed as 
\begin{align*}
    &\mathsf{Pr}\left[g(\mathbf{W})-\mathsf{E}[g(\mathbf{W})]\geq \epsilon\right]\\
    & = \mathsf{Pr}\left[g(\mathbf{W})-\mathsf{E}[g(\mathbf{W})]\geq \epsilon| \mathsf{BD}\right]\mathsf{Pr}[\mathsf{BD}] \\
    &\quad + \mathsf{Pr}\left[g(\mathbf{W})-\mathsf{E}[g(\mathbf{W})]\geq \epsilon|\mathsf{UBD}\right]\mathsf{Pr}[\mathsf{UBD}]\\
    & \leq \exp\left(-\frac{\epsilon^2}{2\kappa^8 \mathsf{E}^2[g(\mathbf{W})]}\right)p_{\kappa}^{tM(I-r)} \\
    &\quad+ \mathsf{Pr}\left[g(\mathbf{W})-\mathsf{E}[g(\mathbf{W})]\geq \epsilon|\mathsf{UBD}\right](1-p_{\kappa}^{tM(I-r)})\\
    &\leq \exp\left(-\frac{\epsilon^2}{2\kappa^8 \mathsf{E}^2[g(\mathbf{W})]}\right)p_{\kappa}^{tM(I-r)} + 1-p_{\kappa}^{tM(I-r)},
\end{align*}
where the last inequality is due to the fact that a conditional probability is always smaller than $1$. Finally, with proper algebraic substitution, we have
\begin{align*}
    &\mathsf{Pr}\left[\Vert(\mathbf{I}_I-\tilde{\mathbf{U}}\tilde{\mathbf{U}}^{\top})\mathbf{X}\Vert_F^2\leq\mathsf{E}[g(\mathbf{W})](1+\epsilon)+\sum_{i\geq r+1}\sigma_i^2\right]\\
    &\geq \left[1-\exp\left(-\frac{\epsilon^2}{2\kappa^8}\right)\right]p_{\kappa}^{tM(I-r)},
\end{align*}
which completes the proof.

\subsection{Proof of Lemma~\ref{Lemma:selection}}\label{Apdx:selection}
First of all, for any $i\leq r$ and $j\geq r+1$, we have $\sigma_i^2-\sigma_{j}^2 \geq \sigma_r^2-\sigma_{r+1}^2$, which gives
\begin{align*}
    &\frac{1}{\tilde{M}}\sum_{i=1}^r\sum_{j\geq r+1}\frac{\lambda_j^2+\lambda_j\mathsf{Tr}(\mathbf{\Lambda})}{\sigma_i^2-\sigma_j^2} \\
    &\leq \frac{r}{\tilde{M}(\sigma_r^2-\sigma_{r+1}^2)}\sum_{j\geq r+1}\lambda_j[\lambda_j+\mathsf{Tr}(\mathbf{\Lambda})].
\end{align*}
Then, define $C =\sum_{k}\mathsf{Tr}\left(\mathbf{X}_k^{\top}\mathbf{X}_k\right) = \sum_{i=1}^I\sigma_i^2$ and $\eta_{\mathrm{th}}=\zeta_{\mathsf{th}} C$. It follows that $\sigma_j^2 \leq \sigma_i^2\leq C/r$, $\forall j\geq r+1>i$, which further gives $r\lambda_j \leq r\sigma_j^2 + \frac{r\sigma^2}{2}\zeta_{\mathsf{th}}C\leq (1+ \frac{r\sigma^2}{2}\zeta_{\mathsf{th}})C$. As a result, one can obtain
\begin{equation*}
    \sum_{j\geq r+1}r\lambda_j[\lambda_j+\mathsf{Tr}(\mathbf{\Lambda})]\leq (1+ \frac{r\sigma^2}{2}\zeta_{\mathsf{th}})C\sum_{j\geq r+1}[\lambda_j+\mathsf{Tr}(\mathbf{\Lambda})],
\end{equation*}
where the summation term can be upper bounded as
\begin{align*}
    &(I-r)\mathsf{Tr}(\mathbf{\Lambda}) + \sum_{j\geq r+1}\lambda_j\\
    &\leq (I-r)C + (I-r)\frac{I\sigma^2}{2}\zeta_{\mathsf{th}}C +  \frac{I-r}{r}C +(I-r)\frac{\sigma^2}{2}\zeta_{\mathsf{th}}C \\
    &\leq \frac{I-r}{(I+1)r}(1 + \frac{r\sigma^2}{2}\zeta_{\mathsf{th}})C.
\end{align*}
Putting the above results together yields
\begin{align*}
    &\frac{1}{\tilde{M}}\sum_{i=1}^r\sum_{j\geq r+1}\frac{\lambda_j^2+\lambda_j\mathsf{Tr}(\mathbf{\Lambda})}{\sigma_i^2-\sigma_j^2}\\
    &\leq \frac{(I-r)C^2}{(\sigma_r^2-\sigma_{r+1}^2)(I+1)r}\frac{(1 + \frac{r\sigma^2}{2}\zeta_{\mathsf{th}})^2}{\tilde{M}},
\end{align*}
where $\frac{(I-r)C^2}{(\sigma_r^2-\sigma_{r+1}^2)(I+1)r}$ can be treated as a constant independent from the variable $\zeta_{\mathsf{th}}$. This completes the proof.

\bibliographystyle{IEEEtran}
\bibliography{Ref}

\end{document}